\documentclass[11pt]{article}
\usepackage{amsthm,amsmath,amsfonts,amssymb}
\usepackage{epsfig}
\usepackage{color}
\usepackage{graphicx}
\newcommand{\R}{{\mathord{\mathbb R}}}

\newcommand{\N}{{\mathord{\mathbb N}}}
\newcommand{\C}{{\mathord{\mathbb C}}}

\def\eps {\varepsilon}

\def\ph  {\varphi}
\newcommand{\chib}{\overline{\chi}}

\newcommand{\BB}{\mathcal{B}}

\newcommand{\FF}{\mathcal{F}}
\newcommand{\HH}{\mathcal{H}}
\newcommand{\LL}{\mathcal{L}}
\newcommand{\RR}{\mathcal{R}}

\newcommand{\WW}{\mathcal{W}}
\newcommand{\UU}{\mathcal{U}}

\renewcommand{\ker}{{\rm Ker}}
\newcommand{\ran}{{\rm Ran}}

\newcommand{\sprod}[2]{\mbox{$\langle #1,#2 \rangle$}}   
\newcommand{\restricted}{|\grave{}\,}
\def\bchi {\boldsymbol{\chi}}
\def\obchi {\overline{\bchi}}
\def\bk {\mathbf{k}}

\newcommand{\ben}{\begin{displaymath}}
\newcommand{\een}{\end{displaymath}}
\newcommand{\beqn}{\begin{equation}}
\newcommand{\eeqn}{\end{equation}}
\newcommand{\beqna}{\begin{eqnarray*}}
\newcommand{\eeqna}{\end{eqnarray*}}

\def\ran{{\rm Ran} \, }
\def\rea{{\rm Re} \, }

\def\Hat{H_{\rm at}}   
\def\Eat{E_{\rm at}}   
\def\Pat{P_{\rm at}}
\def\Patb{\overline{P}_{\rm at}}
\def\inf{{\rm inf}\,}

\def\const{{\rm const}\,}

\def\supp{\operatorname{supp}}

\newcommand{\um}{{\underline{m}}}
\newcommand{\un}{{\underline{n}}}

\newtheorem{lemma}{Lemma}
\newtheorem{theorem}[lemma]{Theorem}

\newtheorem{prop}[lemma]{Proposition}
\newtheorem{cor}[lemma]{Corollary}

\newtheorem{hypothesis}{Hypothesis}

\author{M.~Griesemer and D.~Hasler}

\title{Analytic Perturbation Theory and Renormalization Analysis of Matter
  Coupled to Quantized Radiation}

\author{\vspace{5pt} M. Griesemer$^1$\footnote{
E-mail: marcel@mathematik.uni-stuttgart.de.} and D.
Hasler$^2$\footnote{E-mail: dh8ud@virginia.edu.} \\
\vspace{-4pt} \small{$1.$ Fachbereich Mathematik,
Universit\"at Stuttgart,} \\ \small{D-70569 Stuttgart, Germany}\\
\vspace{-4pt}
\small{$2.$ Department of Mathematics, University of Virginia,} \\
\small{Charlottesville, VA 22904-4137, USA}\\}
\date{}

\begin{document}
\maketitle


\begin{abstract}
For a large class of quantum mechanical models of matter and
radiation we develop an analytic perturbation theory for
non-degenerate ground states. This theory is applicable, for
example, to models of matter with static nuclei and non-relativistic
electrons that are coupled to the UV-cutoff quantized radiation
field in the dipole approximation. If the lowest point of the
energy spectrum is a non-degenerate eigenvalue of the Hamiltonian, we show
that this eigenvalue is an analytic function of the nuclear
coordinates and of $\alpha^{3/2}$, $\alpha$ being the fine structure constant.
A suitably chosen ground state vector depends
analytically on $\alpha^{3/2}$ and it is twice continuously
differentiable with respect to the nuclear coordinates.
\end{abstract}


\section{Introduction}
When a neutral atom or molecule made from static nuclei and
non-relativistic electrons is coupled to the (UV-cutoff) quantized
radiation field, the least point of the energy spectrum becomes
embedded in the continuous spectrum due to the absence of a photon
mass, but it remains an eigenvalue \cite{GLL, LiebLoss2003}. This
\emph{ground state energy} $E$ depends on the parameters of the
system, such as the fine-structure constant, the positions of
static nuclei, or, in the center of mass frame of a translation
invariant model, the total momentum. The regularity of $E$ as a
function of these parameters is of fundamental importance. For
example, the accuracy of the Born-Oppenheimer approximation, a pillar
of quantum chemistry, depends on the regularity of $E$ and on the
regularity of the ground state projection as functions of the
nuclear coordinates. If $E$ were an isolated eigenvalue, like it
is in quantum mechanical description of molecules without
radiation, then analyticity of $E$ with respect to any of the
aforementioned parameters would follow from regular perturbation
theory. But in QED the energy $E$ is not isolated and the analysis
of its regularity is a difficult mathematical problem.

In the present paper we study the problem of regularity, described
above, in a large class of models of matter and radiation where
the Hamiltonian $H(s)$ depends analytically on complex parameters
$s=(s_1,\ldots,s_{\nu})\in \C^{\nu}$ and is defined for values of
$s$ from a complex neighborhood of a compact set
$K\subset\R^{\nu}$. Important properties of $H(s)$ are, that
$H(\bar{s})=H(s)^*$ and that, for $s\in K$, the lowest point,
$E(s)$, of the spectrum of $H(s)$ is a non-degenerate eigenvalue.
Under further assumptions, described below, we show that $E(s)$ and
the projection operator associated with the eigenspace of
$E(s)$ are \emph{real-analytic} functions of $s$ in a neighborhood
of $K$. In particular, they are of class $C^{\infty}$ in this
neighborhood. We apply this result to the  Hamiltonian of a
molecule with static nuclei and non-relativistic electrons that
are coupled to the quantized radiation field in dipole
approximation. By our choice of atomic units, this Hamiltonian
depends on the fine-structure constant $\alpha$ only though a
factor of $\alpha^{3/2}$ in front of the dipole interaction
operator. Hence the role of the parameter $s$ may be played by
$\alpha^{3/2}$ or, after a well-known unitary deformation argument
\cite{Hunziker1986}, by the nuclear coordinates. The general
theorem described above implies that the ground state energy, if
it is a non-degenerate eigenvalue, depends analytically on
$\alpha^{3/2}$ and the nuclear coordinates. The ground state
projection is analytic in $\alpha^{3/2}$, and twice continuously
differentiable with respect to the nuclear coordinates. We remark that the dipole approximation seems necessary for the
analyticity with respect to a power of $\alpha$ \cite{BFP2}.

A further consequence of our main result concerns the accuracy of
the adiabatic approximation to the time evolution $U_{\tau}$
generated by the Schr\"odinger equation
$$
    i\frac{d}{dt}\ph_t = H(t/\tau)\ph_t,\qquad t\in [0,\tau],
$$
in the limit $\tau\to\infty$. If $H(s)$ satisfies the assumptions
of our result mentioned above with $K=[0,1]$, then the ground
state projection $P(s)$ is of class $C^{\infty}([0,1])$ and hence
the adiabatic theorem without gap assumption implies that
$\sup_{t\in[0,\tau]}\|(1-P(t))U_{\tau}(t)P(0)\|=o(1)$ as
$\tau\to\infty$ \cite{Teufel2001, Avron-Elgart1999}. Previously,
in all applications of the adiabatic theorem without gap
assumption the differentiability of $P(s)$ was enforced or
provided by the special form $H(s)=U(s)HU(s)^{-1}$ of $H(s)$ where
$U(s)$ is a unitary and (strongly) differentiable operator
\cite{Avron-Elgart1998, Avron-Elgart1999, Teufel2002b}.

We now describe our main result in detail. We consider a class of
Hamiltonians $H_g(s):D\subset\HH\to\HH$ depending on a parameter
$s\in V$, where $V=\overline{V}$ is a complex neighborhood of some point
$s_0\in\R^{\nu}$. For each $s\in V$,
\begin{equation*}\label{intro1}
   H_g(s) = \Hat(s)\otimes 1 + 1\otimes H_f + g W(s),
\end{equation*}
with respect to $\HH=\HH_{\rm at}\otimes \FF$, where $\HH_{\rm at}$ is an arbitrary
complex Hilbert space and $\FF$ denotes the symmetric Fock space over
$L^{2}(\R^3\times\{1,2\})$. We assume that $\Hat(\bar{s})=\Hat(s)^{*}$ for
all $s\in V$ and that $(\Hat(s))_{s\in V}$ is an analytic family of type
(A). This means that the domain $D$ of $\Hat(s)$ is independent of $s\in V$ and
that $s\mapsto\Hat(s)\ph$ is analytic for all $\ph\in D$. We assume,
moreover, that $\Eat(s_0)=\inf\sigma(\Hat(s_0))$ is a simple and isolated eigenvalue of $\Hat(s_0)$.
The operator $H_f$ describes the energy of the bosons, and $gW(s)$
the interaction of the particle system described by $\Hat(s)$ and the bosons. In
terms of creation and annihilation operators
$$
   H_f = \sum_{\lambda=1,2}\int |\bk| a^{*}_{\lambda}(\bk)a_{\lambda}(\bk)d^3\bk
$$
and
$$
   W(s) = \sum_{\lambda=1,2}\int G_{\bar{s}}(\bk,\lambda)^{*}\otimes a(\bk,\lambda) + G_s(\bk,\lambda)\otimes a^{*}(\bk,\lambda) d^3\bk,
$$
where, $(\bk,\lambda)\mapsto G_s(\bk,\lambda)$, for each $s\in V$ is an element of
$L^2(\R^3\times\{1,2\},\LL(\Hat))$. We assume that $s\mapsto G_s$ is a bounded analytic
function on $V$ and that
\begin{equation}\label{intro4}
     \sup_{s\in V}\sum_{\lambda=1,2}\int\|G_s(\bk,\lambda)\|^2 \frac{1}{|\bk|^{2+2\mu}}d^3\bk <\infty
\end{equation}
for some $\mu>0$. Based on these assumptions we show that a neighborhood $V_0\subset
V$ of $s_0$ and a positive constant $g_0$ exist such that for $s\in
V_0$ and all $g\in [0,g_0)$ the operator $H_g(s)$ has a
non-degenerate eigenvalue $E_{g}(s)$ and a corresponding
eigenvector $\psi_g(s)$ that are both \emph{analytic functions} of
$s\in V_0$. Moreover $E_g(s) = \inf\sigma(H_g(s))$ for $s\in\R\cap
V$. Before commenting on the proof of this result we briefly
review the literature.

In \cite{BFP2} the dependence on $\alpha$ of the ground state and
the ground state energy, $E$, is studied for non-relativistic atoms that
are minimally coupled to the quantized radiation field. It is
shown that $E$ and a suitably chosen ground state vector have
expansions in asymptotic series of powers of $\alpha$ with
$\alpha$-dependent coefficients that may diverge logarithmically
as $\alpha\to 0$. Smoothness is not expected and hence the dipole
approximation seems necessary for our analyticity result. 
Much earlier, in \cite{Froe2}, Fr\"ohlich obtained results on the regularity of
the ground state energy with respect to the total momentum $P$ for the system
of a single quantum particle coupled linearly to a quantized field of massless
scalar bosons. Let $H(P)$ denote the Hamiltonian describing this system at
fixed total momentum $P\in\R^3$. The spectrum of $H(P)$ is of
the form $[E(P),\infty)$ but $E(P)$ is not an eigenvalue for
$P\neq 0$ \cite{Froe2} (see \cite{HasAbs} for similar results on positive
ions). For a non-relativistic particle of mass $M$, Fr\"ohlich shows that 
$P\mapsto E(P)$ is differentiable a.e. in $\{|P|<(\sqrt{3}-1)M\}$, and that $\nabla
E(P)$ is locally Lipschitz \cite[Lemma 3.1]{Froe2}.
This work was recently and independently continued by Alessandro Pizzo and Thomas Chen
for systems with a fixed ultraviolet cutoff \cite{Pizzo2003, Chen2001}. After a unitary, $P$-dependent transformation of
the Hamiltonian $H(P)$, Pizzo obtains a ground state vector
$\phi^{\sigma}(P)$ that is H\"older continuous with respect to $P$
uniformly in an infrared cutoff $\sigma>0$. Chen studied the regularity of
$E(P,\sigma)$ for a non-relativistic particle coupled minimally to the
quantized radiation field \cite{Chen2001}. He estimates $|\partial^{\beta}_{|p|}(E(|p|,\sigma)-p^2/2)|$
uniformly in $\sigma>0$ for $\beta\leq 2$.
He also asserts that $E(p,\sigma)$ is of class $C^2$ even for $\sigma=0$. In
\cite{Hunziker1986} Hunziker proves analyticity with respect to the nuclear
coordinates for non-relativistic molecules \emph{without} radiation. The ground state energy is isolated but the
Hamiltonian is \emph{not analytic} with respect the nuclear coordinates. It only
becomes analytic after a suitable unitary deformation (see
Section~\ref{sec:dipole}) introduced by Hunziker.

The results of the present paper are derived using the
renormalization technique of Bach et al.~\cite{bacchefrosig:smo,
bacfrosig:ren}, in a new version that we take from \cite{FGSig}.  
Like the authors of \cite{FGSig} we use a simplified renormalization map that
consists of a Feshbach-Schur map and a scaling transformation only.
In the corresponding spectral analysis
the Hamiltonian is diagonalized, with
respect to $H_f$, in a infinite sequence of renormalization steps.
In each step the off-diagonal part becomes smaller,
thanks to \eqref{intro4}, and the spectral parameter is adjusted
to enforce convergence of the diagonal part. This method provides
a fairly explicit construction of an eigenvector of $H(s)$, even
for complex $s$, where $H(s)$ is not self-adjoint. We show first,
that the parameters of the renormalization analysis can be chosen
independent of $s$ and $g$ in neighborhoods of $s=s_0$ and $g=0$,
second, that all steps of the renormalization analysis preserve
analyticity, and third, that all limits taken are uniform in $s$,
which implies analyticity of the limiting functions. On a
technical level, these three points are the main achievements
of this paper.

It seems unlikely that another approach, not based on a
renormalization analysis would yield a result similar to ours. The
proof of analyticity requires the construction of an eigenvector
for complex values of $s$ where $H(s)$ is not self-adjoint and
hence, variational methods, for example, are not applicable. There is, of
course, the tempting alternative approach to first introduce a
positive photon mass $\sigma$ to separate the least energy from
the rest of the spectrum. But the neighborhood of analyticity
obtained in this way depends on the size of $\sigma$ and vanishes
in the limit $\sigma\to 0$.

We conclude this introduction with a description of the
organization of this paper along with the strategy of our proof.

In Section \ref{sec:model} we introduce the class of Hamiltonians
$(H(s))_{s\in V}$, we formulate all hypotheses, and state the main
results. In Section~\ref{sec:dipole} they are applied to
non-relativistic QED in dipole approximation to prove our results
mentioned above on regularity with respect to $\alpha^{3/2}$ and
the nuclear coordinates.

Section~\ref{sec:fesh} describes the smoothed Feshbach transform
$F_{\chi}(H)$ of an operator $H$ and the isomorphism $Q(H)$
between the kernels of $F_{\chi}(H)$ and $H$ (isospectrality of the Feshbach
transform). The transform $H\mapsto F_{\chi}(H)$ was discovered in \cite{bacchefrosig:smo}, and generalized to the
form needed here in \cite{grihas:smo}.

In Section~\ref{sec:effect} we perform a first Feshbach
transformation on $H(s)-z$ to obtain an effective Hamiltonian
$H^{(0)}[s,z]$ on $\HH_{\rm red}=P_{[0,1]}(H_f)\FF$, $\FF$ being the
Fock-space. We show that $H^{(0)}(s,z)$ is analytic in $s$ and
$z$. By the isospectrality of the Feshbach transform, the eigenvalue problem
for $H(s)$ is now reduced to finding a value $z(s)\in\C$ such that $H^{(0)}(s,z(s))$
has a nontrivial kernel.

In Section \ref{sec:renorm} we introduce a Banach space $\WW_{\xi}$ and a
linear mapping $H:\WW_{\xi}\to\LL(\HH_{\rm red})$. The renormalization transformation
$\RR_{\rho}$ is defined on a polydisc $\BB(\rho/2,\rho/8,\rho/8)\subset
H(\WW_{\xi})$ as the composition of a Feshbach transform and a rescaling
$\bk\mapsto\rho\bk$ of the photon momenta $\bk$. $\rho\in (0,1)$ is the factor
by which the energy scale is reduced in each renormalization step. $\RR_{\rho}$ takes
values in $\LL(\HH_{\rm red})$ and, like the Feshbach transform, it is
isospectral.

In Section \ref{sec:analyt-fesh} it is shown that the analyticity
of a family of Hamiltonians is preserved under the renormalization
transformation. This is one of the key properties on which our strategy is based.

Sections~\ref{sec:iterate} and \ref{sec:ev} are devoted to the solution of our
spectral problem for $H^{(0)}(s,z)$ by iterating the Renormalization
map. Since this procedure is pointwise in $s$ with estimates that hold
uniformly on $V$, we drop the parameter $s$ for notational simplicity.

In Sections~\ref{sec:iterate} we define $H^{(n)}[s,z]=\RR^n H^{(0)}[s,z]$ for values of the spectral parameter $z$ from non-empty sets
$U_n(s)$. These sets are nested, $U_n(s)\supset U_{n+1}(s)$, and they shrink to a
point, $\cap_{n} U_n(s)=\{z_{\infty}(s)\}$. Since $H^{(n)}(z_{\infty})\to
\const H_f$ as $n\to\infty$ in the norm of $\LL(\HH_{\rm red})$ and since the vacuum
$\Omega$ is an eigenvector $H_f$ with eigenvalue zero, it follows, by the isospectrality of $\RR$, that zero is an eigenvalue
of $H^{(n)}(z_{\infty})$ for all $n$.

In Section~\ref{sec:ev} a vector $\ph_n$ in the kernel of
$H^{(n)}(z_{\infty})$ is computed by compositions of scaling
transformations and mappings $Q(H^{(k)}(z_{\infty}))$, $k\geq n$, applied to $\Omega$. $\ph_{gs}=Q(H^{(0)}(z_{\infty}))\ph_0$
is an eigenvector of $H$ with eigenvalue $z_{\infty}$.

In Section~\ref{sec:analyt-ev} we show that $s\mapsto
z_{\infty}(s)$ is analytic and that $Q(H^{(n)}(z_{\infty}(s)))$ maps analytic
vectors to analytic vectors. Since the vacuum $\Omega$ is trivially analytic
in $s$, it follows that $\ph_{gs}(s)$ is analytic in $s$.

In the Appendices~\ref{sec:neighbor} and \ref{sec:aux} we collect technical auxiliaries and for
completeness we give a proof of $H^{(n)}(z_{\infty})\to
\const H_f$ as $n\to\infty$, although this property is not used explicitly.
\\

\noindent \emph{Acknowledgment.} We thank J\"urg Fr\"ohlich and
Israel Michael Sigal for numerous discussions on the
renormalization technique.  
M.G. also thanks Volker Bach for explaining the results of \cite{BFP2}, and
Ira Herbst for the hospitality at the University of
Virginia, were large parts of this work were done.


\newpage
\section{Assumptions and Main Results}
\label{sec:model}

We consider families of (unbounded) operators
$H(s):D(H(s))\subset\HH\to\HH$, $s\in V$, where $V\subset\C^{\nu}$ is open,
symmetric with respect to complex conjugation and $V\cap\R^{\nu}\neq\emptyset$.
The Hilbert space $\HH$ is a tensor product
$$
    \HH = \HH_{\rm at}\otimes\FF,\qquad
    \FF = \bigoplus_{n=0}^\infty S_n(\otimes^n\mathfrak{h}),
$$
of an arbitrary, separable, complex Hilbert space $\HH_{\rm at}$
and the symmetric Fock space $\FF$ over the Hilbert space $\mathfrak{h}:= L^2(\R^3\times\{1,2\};\C)$
with norm given by
$$
  \|h\|^2 :=\sum_{\lambda=1,2}\int|h(\mathbf{k},\lambda)|^2
    d^3\mathbf{k},\qquad h\in \mathfrak{h}.
$$
Here $S_0(\otimes^0\mathfrak{h}):=\C$ and for $n\geq 1$,
$S_n\in\LL(\otimes^n\mathfrak{h})$ denotes the
orthogonal projection onto the subspace left invariant by all
permutation of the $n$ factors of $\mathfrak{h}$. To simplify our notation we
set
$$
  k:=(\mathbf{k},\lambda),\qquad\int dk :=
  \sum_{\lambda=1,2}\int\,d^3\mathbf{k},\qquad |k|:=|\mathbf{k}|,
$$
throughout the rest of this paper.

For each $s\in V$, the operators $H(s)$ is a sum
\begin{equation}\label{ham}
   H_g(s) = \Hat(s)\otimes 1 + 1\otimes H_f + g W(s),
\end{equation}
of a closed operator $\Hat(s)$ in $\HH_{\rm at}$, the second
quantization, $H_f$, of the operator $\omega$ on $L^2(\R^3\times\{1,2\})$ of
multiplication with
$$
     \omega(k) = |k|,
$$
and an interaction operator $gW(s)$, $g\geq 0$ being a coupling
constant. The operator $W(s)$ is the sum
$$
    W(s) = a(G_{\bar{s}}) + a^{*}(G_s)
$$
of an annihilation operator, $a(G_{\bar{s}})$, and a creation
operator, $a^{*}(G_s)$, associated with an operator $G_s \in
\LL(\HH_{\rm at},\HH_{\rm at}\otimes\mathfrak{h})$. The creation operator,
$a^{*}(G_{s})$, as usual, is defined as the closure of the linear
operator in $\HH$ given by
$$
a^*(G_s)(\varphi \otimes \psi ):=\sqrt{n}S_{n}(G_s \varphi \otimes
\psi),
$$
if $\varphi\in \HH_{\rm at}$ and $\psi \in
S_{n-1}(\otimes^{n-1}\mathfrak{h})$. The annihilation operator $a(G_{s})$ is the
adjoint of $a^{*}(G_{s})$.

Hypotheses~\ref{hyp:G} below will imply that $H_g(s)$ is well defined on
$D(\Hat(s))\otimes D(H_f)$ and closable. To formulate it,
some preliminary remarks are necessary. Let $L^2(\R^3,\LL(\HH_{\rm
at}))$ be the Banach space of (weakly) measurable and square integrable
functions from $\R^3$ to $\LL(\HH_{\rm
  at})$. Every element $T$ of this space defines a linear operator $T:\HH_{\rm
  at}\to L^2(\R^3,\HH_{\rm at})$ by
$$
   (T\ph)(k) := T(k)\ph.
$$
This operator is bounded and $\|T\|\leq \|T\|_2$. Since $L^2(\R^3,\HH_{\rm
  at})\simeq \HH_{\rm at}\otimes\mathfrak{h}$, we may consider $T$ as an element
of $\LL(\HH_{\rm at},\HH_{\rm at}\otimes\mathfrak{h})$ and hence
$L^2(\R^3,\LL(\HH_{\rm at}))$ as a subspace embedded in
$\LL(\HH_{\rm at},\HH_{\rm at}\otimes\mathfrak{h})$.
\begin{hypothesis} \label{hyp:G}
\emph{The mapping $s \mapsto G_s$ is an bounded analytic function on $V$ with
values in $L^2(\R^3,\LL(\HH_{\rm at}))$, and there exists a $\mu>0$, such that
\begin{equation*}
    \sup_{s \in V} \int\frac{1}{|k|^{2+2\mu}} \|G_s(k)\|^2\, dk <\infty.
\end{equation*}}
\end{hypothesis}

By Lemma~\ref{lm:A1}, $\|a^{\#}(G_s)(H_f+1)^{-1/2}\|\leq
\|G_s\|_{\omega}$, where
$$
     \|G_s\|_{\omega}^2:= \int_{\R^3}\left\|G_s(k) \right\|^2 (|k|^{-1}+1)\,
     dk.
$$
Hence Hypothesis~\ref{hyp:G} implies that $W(s)$ and $W(s)^*$
are well defined on $\HH_{\rm at}\otimes D(H_f)$. It follows that
$H_g(s)$ is defined on $D(\Hat(s))\otimes D(H_f)$ and that the
adjoint of this operator has a domain which contains
$D(\Hat(s))^{*}\otimes D(H_f)$. This subspace is dense because
$\Hat(s)$ is closed. That is, $H_g(s):D(\Hat(s))\otimes
D(H_f)\subset\HH\to \HH$ has a densely defined adjoint, and hence it
is closable.

\begin{hypothesis}\label{hyp:H}
\emph{\begin{itemize}
\item[(i)] $\Hat(s)$ is an analytic family of operators in the
sense of Kato and $\Hat(s)^{*}=\Hat(\bar{s})$ for all $s\in V$. In particular, $\Hat(s)$
is self-adjoint for $s\in \R^{\nu}\cap V$.
\item[(ii)] There exists a point $s_0\in V\cap\R^{\nu}$ such that
$\Eat(s_0):=\inf\sigma(\Hat(s_0))$ is an isolated, non-degenerate
eigenvalue of $\Hat(s_0)$.
\end{itemize}}
\end{hypothesis}

For the notion of an analytic family of operators in the sense of Kato we
refer to \cite{reesim:ana}. The definition given there readily generalizes to
several complex variables. We recall that a function of several complex
variables is called analytic if it is analytic in each variable separately.

By Hypothesis~\ref{hyp:H}, (ii), and the Kato-Rellich theorem of analytic
perturbation theory \cite{reesim:ana}, there is exactly one point $E_{\rm at}(s)$ of
$\sigma(\Hat(s))$ near $E_{\rm at}(s_0)$, for $s$ near $s_0$, and this point is
a non-degenerate eigenvalue of $\Hat(s)$. Moreover, for $s$ near  $s_0$, there
is an analytic projection onto the eigenvector of  $E_{\rm at}(s)$, which is given by
$$
P_{\rm at}(s) := \frac{1}{2\pi i} \int_{|E_{\rm
at}(s)-z|=\epsilon} \big(z-H_{\rm at}(s)\big)^{-1}\, dz,
$$
for $\epsilon >0$ sufficiently small. We set $ \overline{P}_{\rm at}(s) = 1 -
P_{\rm at}(s)$.
\begin{hypothesis} \label{hyp:R}
\emph{Hypothesis \ref{hyp:H} holds and there exists a neighborhood
$\UU \subset V\times\C$ of $(s_0,\Eat(s_0))$ such that for all
$(s,z) \in \UU$, $|E_{\rm at}(s) - z | < 1/2$ and
$$
   \sup_{(s,z)\in \UU} \sup_{q \geq 0}
   \left\|\frac{q+1}{\Hat(s) - z + q} \overline{P}_{\rm at}(s)\right\| <
\infty.
$$}
\end{hypothesis}

\noindent \emph{Remarks.}
\begin{enumerate}
\item Hypothesis \ref{hyp:R} is satisfied, e.g., if Hypothesis~\ref{hyp:H}
 holds and $\Hat(s)$ is an analytic family of type~(A), see Corollary~\ref{cor:main2}.
\item  The condition  $|E_{\rm at}(s) - z | <
1/2$ is needed in the proof of Theorem~\ref{thm:fesh} and related to
the constant $3/4$ in the construction of $\chi$. Since
$s\mapsto\Eat(s)$ is continuous, it can always be met by choosing
$\UU$ sufficiently small. However, the smaller we choose $\UU$ the
smaller we will have to choose the coupling constant $g$. Optimal
bounds on $g$ could possibly be obtained by scaling the operator
such that the gap in  $H_{\rm at}$ is comparable to one.
\end{enumerate}

We are now ready to state the main results.

\begin{theorem}\label{thm:main}
Suppose Hypotheses \ref{hyp:G}, \ref{hyp:H} and \ref{hyp:R} hold.
Then there exists a neighborhood $V_0\subset V$ of $s_0$ and  a
positive constant $g_0$ such that for all $s\in V_0$ and all $g<g_0$
the operator $H_g(s)$ has an eigenvalue $E_g(s)$ and
a corresponding eigenvector $\psi_g(s)$ that are both analytic
functions of $s\in V_0$ such that
$$
     E(s)=\inf\sigma(H(s))
$$
for $s\in V_0\cap\R^{\nu}$.
\end{theorem}

\noindent \emph{Remark.} For $s\in V_0\cap\R^{\nu}$ and $g$ sufficiently small, the
eigenvalue $E_g(s)$ is non-degenerate by
Hypothesis~\ref{hyp:H}~(ii) and a simple overlap estimate
\cite{BFS1}.

\begin{cor}\label{cor:main1}
Assume Hypotheses \ref{hyp:G} and \ref{hyp:H} are satisfied and
that there exists a $C$ such that for all $s\in V$
\begin{equation}\label{semi-bound}
   \rea\sprod{\ph}{\Hat(s)\ph} \geq -C\sprod{\ph}{\ph},\qquad
   \text{for}\ \ph\in D(\Hat(s)).
\end{equation}
Then the conclusions of Theorem~\ref{thm:main} hold.
\end{cor}

\begin{proof}
It suffices to verify Hypothesis~\ref{hyp:R}, then the corollary will follow
from Theorem~\ref{thm:main}.
For all $s\in V$, $z\in \C$ with $|z-\Eat(s_0)|\leq 1$,
$q\geq q^{*}:=C+|\Eat(s_0)|+2$, and $\ph\in D(\Hat(s))$ with
$\|\ph\|=1$,
\begin{eqnarray*}
\|(\Hat(s)-z+q)\ph\| &\geq & \rea\sprod{\ph}{(\Hat(s)-z+q)\ph}\\
  &\geq & q-C-|\Eat(s_0)| -1\ \geq\ 1.
\end{eqnarray*}
Since $\Hat(s)^*=\Hat(\bar{s})$ an analog estimate holds for
$\Hat(s)^*$. This proves that $B_1(\Eat(s_0))\subset \rho(\Hat(s)+q)$
for $s\in V$, $q\geq q^*$, and that
\begin{equation}\label{cor-main1}
  \sup_{s\in V, |z-\Eat(s_0)|\leq 1, q\geq q^*}
  \left\|\frac{q+1}{\Hat(s)-z+q}\right\|\leq \frac{q^*}{q^{*}-C+|\Eat(s_0)|+1}.
\end{equation}

We now turn to the case where $0\leq q\leq q^*$. The set
$$
  \Gamma:=\{(s,z)\in \C^{\nu}\times \C|
  z\in \rho(H(s)\upharpoonright\overline{P}_{\rm at}(s)\HH)\}
$$
is open and $(\Hat(s)-z)^{-1}\overline{P}_{\rm at}(s)$ is analytic
on $\Gamma$ \cite{reesim:ana}. On the other
 hand
$$
   \gamma:=\{(s_0,\Eat(s_0)-q)| 0\leq q\leq q^*\}
$$
is a compact subset of $\Gamma$. It follows that the distance
between $\gamma$ and the complement of $\Gamma$ is positive. Thus
if $g$ and $\delta>0$ are small enough, then
$$
   \{(s,z-q): |s-s_0|\leq \delta,\ |z-\Eat(s_0)|\leq \delta,\ 0\leq q\leq q^*\}
$$
is a compact subset of $\Gamma$ on which
$(H(s)-z)^{-1}\overline{P}_{\rm at}(s)$ is uniformly bounded.
Comparing with \eqref{cor-main1} we conclude that for $\delta<1$
so small that $B_{\delta}(s_0)\subset V$ the Hypothesis
\ref{hyp:R} holds with $U=B_{\delta}(s_0)\times
B_{\delta}(E(s_0))$
\end{proof}

The following corollaries prove the assertions in the introduction.

\begin{cor}\label{cor:main2}
Suppose Hypothesis~\ref{hyp:G} holds and let $\Hat(s)$ be an
analytic family of type (A) with $H_{\rm at}(s)^* = H_{\rm
at}(\bar{s})$ for all $s \in V$. If $E(s_0)=\inf\sigma(H(s_0))$ is a
non-degenerate isolated eigenvalue of $\Hat(s_0)$, then the
conclusions of Theorem~\ref{thm:main} hold.
\end{cor}

\begin{proof}
By Corollary~\ref{cor:main1} it suffices to show that
\eqref{semi-bound} holds. To this end we set
$$
   T(s):= H_{\rm at}(s) - E_{\rm at}(s_0)
$$
and $R:=(T(s_0)+1)^{-1}$. Since $T(s)$ is an analytic family of
type (A), the operators $T(s)R$ and $RT(s)$ are bounded and weakly
analytic, hence strongly analytic \cite{kat:per}. It follows that
$(T(s)-T(s_0))R\to 0$ and $R(T(s)-T(s_0))\to 0$ as $s\to 0$. By
abstract interpolation theory
$$
     R^{1/2}(T(s)-T(s_0))R^{1/2}\to 0\qquad s\to s_0.
$$
We choose $\eps>0$ so that $B_{\eps}(s_0)\subset V$ and
$$
    \sup_{|s|<\eps}\|R^{1/2}(T(s)-T(s_0))R^{1/2}\| \leq 1/2.
$$
It follows that $|\sprod{\ph}{[T(s)-T(s_0)]\ph}|\leq
(1/2)\sprod{\ph}{(T(s_0)+1)\ph}$ and hence that
\begin{eqnarray*}
   \rea\sprod{\ph}{T(s)\ph} &\geq & \sprod{\ph}{T(s_0)\ph} -
   |\sprod{\ph}{[T(s)-T(s_0)]\ph}|\\
   &\geq & \frac{1}{2}\sprod{\ph}{T(s_0)\ph}
   -\frac{1}{2}\sprod{\ph}{\ph}\ \geq\
   -\frac{1}{2}\sprod{\ph}{\ph}
\end{eqnarray*}
which proves \eqref{semi-bound}.
\end{proof}

\emph{ Remark.} Embedded eigenvalues generically disappear under perturbations.
However, if a non-degenerate eigenvalue persists, one
might conjecture that the eigenvalue and a suitable eigenvector are analytic
functions of the perturbation parameter, provided the Hamiltonian is analytic
in this parameter. This conjecture is \emph{wrong}, as the following
example shows. Let $\HH=L^2(\R)\oplus\C$, and let
$$
H(s) = \left( - \frac{d^2}{dx^2} + s V \right)\oplus 0  ,
$$
where $V$ denotes the characteristic function of the interval $[-1,1]$.
Then $H(s)$, $s\in\C$, is an analytic family of type ($A$) and
for $s\in\R$, $E(s)={\rm inf} \sigma(H(s))$ is a non-degenerate eigenvalue.
But $E(s)$ is not analytic because
\begin{align*}
& E(s) < 0 , \qquad s < 0 \\
& E(s) = 0 , \qquad s \geq 0  .
\end{align*}
The corresponding eigenvector is not even continuous at $s=0$.

\begin{cor}\label{cor:main3}
Suppose the assumptions of Theorem~\ref{thm:main}, Corollary~\ref{cor:main1}
or Corollary~\ref{cor:main2} are satisfied for all $s_0$ of a compact set
$K\subset V\cap\R^{\nu}$. Then there exists a neighborhood $V_0\subset V$
of $K$ and a positive $g_0$ such that for all $s\in V_0$ and all $g<g_0$ there
is an analytic complex-valued function $E_g$ and an analytic projection-valued
functions $P_g$ on $V_0$, such that
$$
  H_g(s)P_g(s) = E_g(s)P_g(s), \qquad \text{for}\ s\in V_0,
$$
and $E_g(s)=\inf\sigma(H_g(s))$ for $s\in V_0\cap \R^{\nu}$.
\end{cor}

\begin{proof}
By the compactness of $K$ there exist open sets $V_1\dots,
V_N\subset\C^{\nu}$ and positive numbers $g_1,\dots,g_N$ provided
by Theorem~\ref{thm:main}, such that
$$
   K\subset \bigcup_{j=1}^{N} V_{j}.
$$
Let $E_j(s)$ and $\psi_j(s)$ be the
corresponding eigenvalues and eigenvectors defined for $s\in V_j$ and $g<g_j$.
We may assume that $\sprod{\psi_{j}(\bar{s})}{\psi_j(s)}\neq 0$
for all $s\in V_j$ and all $j$. Then the operators $P_j(s):\HH\to\HH$ defined
by
$$
   P_j(s)\ph = \frac{\sprod{\psi_{j}(\bar{s})}{\ph}}{\sprod{\psi_{j}(\bar{s})}{\psi_j(s)}}\psi_j(s)
$$
are analytic functions of $s\in V_j$, $P_j(s)^2=P_j(s)$, and $H_g(s)P_j(s) =
E_j(s)P_j(s)$. We choose $g_0\leq \min\{g_1,\dots,g_N\}$ so small, that all
eigenvalues $E_1(s),\dots,E_N(s)$ are non-degenerate for real $s$ and $g\in[0,g_0)$, and we
define $V_0:=\cup_{j=1}^N V_j$. Then for $g<g_0$ and $s\in V_i\cap V_j\cap\R^{\nu}$,
\begin{equation}\label{eq:ViVj}
    E_i(s)=E_j(s),\qquad P_i(s)=P_j(s),
\end{equation}
and hence, by analyticity, \eqref{eq:ViVj} must hold for all  $s\in V_i\cap
V_j$. This proves that $E_g(s)$ and $P_g(s)$ are well-defined on $V_0$ by
$E_g(s):=E_j(s)$ and $P_g(s):=P_j(s)$ for $s\in V_j$, and have the desired
properties.
\end{proof}


\section{Non-Relativistic QED in Dipole-Approximation}
\label{sec:dipole}

In this section we apply Theorem~\ref{thm:main} to the Hamiltonians describing
molecules made from static nuclei and non-relativistic electrons coupled to
the UV-cutoff quantized radiation field in dipole approximation.
For justifications of this model see \cite{Gri2006, Cohen1}.

A (pure) state of system of $N$ spinless electrons and transversal photons is
described by a vector in the Hilbert space $\HH=L_a(\R^{3N};\C)\otimes \FF$, where
$L_a^2(\R^{3N};\C)$ denotes the space of square integrable,
antisymmetric functions of $(x_1,\ldots,x_N)\in\R^{3N}$, and $\FF$
is the symmetric Fock space over $L^2(\R^3\times\{1,2\})$. We choose
units where $\hbar$, $c$, and four times the Rydberg energy are
equal to one, and we express all positions in multiples of one half
of the Bohr-radius, which, in our units, agrees with the
fine-structure constant $\alpha$. In these units the Hamiltonian
reads
\begin{eqnarray*}
H(X,\alpha) &=& H_{\rm at}(X)  + \alpha^{3/2} \sum_{j=1}^N g(x_j)
x_j \cdot E(0) + H_f \; ,
\end{eqnarray*}
with
\begin{eqnarray*} H_{\rm at}(X)
&:=&\sum_{j=1}^N (-\Delta_j) + \sum_{j<l} \frac{1}{|x_l - x_j|} -
\sum_{j,k} \frac{ Z_j}{|x_j - X_k|}
\end{eqnarray*}
where $Z_1,\ldots,Z_{K}\in \N$ denote atomic numbers, and
$$
E(0) = \sum_{\lambda=1,2} \int_{|\bk| < \Lambda} d\bk  \sqrt{|\bk|}  i
\varepsilon(\bk,\lambda) ( a^{*}(\bk,\lambda) - a(\bk,\lambda)) \; ,
$$
is the quantized electric field evaluated at the origin $0\in\R^3$.
The ultraviolet cutoff $\Lambda>0$ is an arbitrary but finite constant, the
polarization vectors $\eps(\bk,1)$ and $\eps(\bk,2)$ are unit vectors in
$\R^3$ that are orthogonal to each other and to $\bk$, and
$g \in C_0^\infty(\R^3)$ is a space-cutoff with $g\equiv 1$ on an open ball $B \subset \R^3$ containing the positions
$X_1,\ldots,X_K$ of the nuclei. The following theorem is a consequence of Corollary
\ref{cor:main2}.

\begin{theorem} Suppose ${\rm inf} \, \sigma( H_{\rm at}(X))$ is a
non-degenerate and isolated eigenvalue of $H_{\rm at}(X)$. Then in a
neighborhood of $\alpha=0$ the ground state energy and a suitably chosen ground
state vector are real-analytic functions of $\alpha^{3/2}$.
\end{theorem}

\begin{proof}
For $s\in V:= B_1(0)$ we define $H_g(s)$ by the operator \eqref{ham} with
$$
    G_s(\bk,\lambda) :=
    s\sum_{j=1}^N\sqrt{|\bk|}\chi(|\bk|\leq\Lambda)i\eps(\bk,\lambda)\cdot x_j g(x_j)
$$
so that $H_g(\alpha^{3/2}/g)=H(X,\alpha)$. The Hamiltonian $H_{\rm
at}(X)$ is trivially analytic of type ($A$) in $V$ and $G_s$ satisfies
Hypothesis~\ref{hyp:G} with, for example, $\mu=1/2$. Hence the
conclusions of Theorem~\ref{thm:main} holds by
Corollary~\ref{cor:main2}. This proves the theorem.
\end{proof}

The next theorem concerns the regularity of $\inf\sigma(H(X,\alpha))$ with
respect to the nuclear coordinates $X=(X_1,\dots,X_K)\in \R^{3K}$.

\begin{theorem} Suppose  ${\rm inf} \, \sigma(H_{\rm at}(X))$ is a non-degenerate
and isolated eigenvalue of $H_{\rm at}(X)$, where $X\in B^K$ and
$X_r \neq X_s$ for $r \neq s$. Then for $\alpha$ sufficiently small,
there exists a neighborhood $U$ of $0\in\R^{3K}$ such that:
\begin{itemize}
\item[(a)] For each $\xi\in U$, $E(\xi)={\rm inf} \sigma (H(X+\xi))$ is an  eigenvalue
  of $H(X+\xi)$ and a real-analytic function of $\xi$.
\item[(b)] There is an eigenvector belonging to $E(\xi)$, which is of class
  $C^2$ with respect to $\xi$.
\end{itemize}
\end{theorem}

\noindent\emph{Remark.} The operators $H_{\rm at}(X)$ do not form an analytic family in
the sense of Kato and hence Theorem~\ref{thm:main} is not immediately
applicable. This problem is circumvented by a well-known deformation argument \cite{Hunziker1986}.

\begin{proof}
By assumption on $X_1,\dots,X_K$, we can find functions $f_1,\ldots,f_K \in
C_0^\infty(\R^3)$, with $\supp(f_r)\subset B$ and $f_r(X_s) = \delta_{rs}$. For each $\xi =
(\xi_1,\ldots,\xi_K) \in \R^{3K}$ we define a vector field $v_\xi$ on $\R^3$, by
$$
v_\xi(x) = \sum_{r=1}^K \xi_r f_r(x) \; ,
$$
It is not hard to see that for small $\xi$ the map
$$\phi_\xi : (x_1,...,x_N) \mapsto (x_1 +
v_\xi(x_1),...,x_N+v_\xi(x_N))$$ is a diffeomorphism of $\R^{3N}$
\cite{Hunziker1986}. Moreover,
$$
U_{\xi} \psi  := |D\phi_\xi|^{1/2} ( \psi \circ \phi_\xi ) \; ,
$$
defines a unitary transformation $U_{\xi}$ on $\HH$. A
straightforward calculation shows that, for real and small $\xi$,
$$
 \widetilde{H}(\xi) := U_{\xi} H(X + \xi) U_{\xi}^{-1} = H_{\rm at}(X;\xi) + W(\xi) + H_f \;
,
$$
with
\begin{align} \label{eq:wxi}
W(\xi) &= \alpha^{3/2} \sum_{j=1}^N g(x_j) ( x_j + v_\xi(x_j)) \cdot
E(0) \\
H_{\rm at}(X;\xi) &= T_\xi - \sum_{r,j} Z_k V_{\xi}(x_j , X_r) +
 \sum_{j<l} V_\xi(x_j,x_l) \nonumber
 \end{align}
where
 \begin{align*}
 T_\xi &= U_\xi \sum_{j=1}^N (-\Delta_j) U_\xi^{-1} \\
V_\xi(x,y) &= | x - y + v_\xi(x) - v_\xi(y) |^{-1}
 \; .
\end{align*}
In \eqref{eq:wxi} we used that $g(x+ v_\xi(x))
=g(x)$, by the smallness of $\xi$. In \cite{Hunziker1986} it is proven using standard estimates that $H_{\rm
at}(X,\xi)$ has an extension to $\xi\in\C^{3K}$ and this
extension is an analytic family of type (A) for $\xi$ in a
neighborhood of zero. One easily verifies that $W(\xi)$ satisfies
Hypothesis I. It follows that Corollary \ref{cor:main2} is
applicable and thus, for small $\alpha$, $\widetilde{H}(\xi)$ has an
eigenvalue $E(\xi)$ with eigenvector $\varphi_\xi$, both analytic in
$\xi$, and $\varphi_\xi$ is a ground state for real $\xi$.  Since,
for small and real $\xi$, $H(\xi + X)$ is unitarily equivalent to
$\widetilde{H}(\xi)$, (a) follows.

To prove (b) we show that $\xi \mapsto U_\xi^{-1} \varphi_\xi$ is a
$C^2$ function in a neighborhood of zero. Let $S_{\xi} :=
U_\xi^{-1}$. Throughout the proof, with the exception of Step~3, we
assume that $\xi$ is real. Using dominated convergence, one sees
that $\xi \mapsto S_{\xi} \varphi$ is continuous for $\varphi\in \HH
\cap (C_0^\infty(\R^{3N})\otimes\FF)$. Since such functions constitute
a dense subset of $\HH$   and $S_{\xi}$ is uniformly bounded, it
follows that $\xi \mapsto S_{\xi}$ is strongly continuous.
 We shall adopt the following conventions in this
proof: the labels  $\alpha, \beta$ run over the set
$\{1,\ldots,K\}\times\{1,2,3\}$, and $\partial_\beta = \partial /
\partial \xi_\beta$ with  $\xi_\beta = (\xi_j)_s$ for $\beta=(j,s)$; the labels $\mu, \nu$ run
over the set $\{1,\ldots,N\}\times\{1,2,3\}$ and  $p_\mu =
(p_l)_s$ for $\mu=(l,s)$. \\

\noindent \underline{Step 1}: If $\psi \in D(|p|)$, then $\xi
\mapsto S_{\xi} \psi$ is  $C^1$ and for all $\beta$, $\partial_\beta
S_{\xi} \psi = S_{\xi} A_\beta(\xi) \psi$ with
\begin{eqnarray}
 [A_\beta(\xi) \psi](x) &=&   | D \phi_\xi(x)|^{1/2} \frac{d}{ds}|
D\phi_{\xi + se_\beta}^{-1} (\phi_\xi(x)) |^{1/2} \Big|_{s=0} \psi(x)  \nonumber \\
&& + \sum_{j=1}^N  \frac{d}{ds}  {\phi_{\xi + se_\beta }^{-1}}
(\phi_\xi(x))_j \Big|_{s=0} \cdot (\nabla_j \psi)(x) \; ,
\label{eq:infgen}
\end{eqnarray}
where  $e_\beta \in \R^{3K}$ denotes the unit vector with components
$(e_\beta)_\gamma = \delta_{\beta,\gamma}$.\\

For $h_1, h_2 \in C_0^\infty(\R^{3N}) \otimes \FF$,
we calculate the partial derivative using $S_{\xi} = U_\xi^{-1}$,
\begin{eqnarray}
\partial_\beta \sprod{ h_1 }{ S_{\xi} h_2 } = \partial_\beta \sprod{ h_1 }{ |D
\phi_\xi^{-1} |^{1/2} ( h_2 \circ \phi_\xi^{-1} ) }  =  \sprod{ h_1
}{ S_{\xi} A_\beta(\xi) h_2 } \; , \label{eq:weakder555}
\end{eqnarray}
where in the second equality we used the product rule of
differentiation and the identity $|D\phi_\xi^{-1}(x) | | D \phi_\xi
(\phi_\xi^{-1}(x) ) | = 1$. Integrating  \eqref{eq:weakder555}, we
find
\begin{eqnarray}
\sprod{ h_1 }{ V_{\xi + t e_\beta} h_2 } &=&
  \sprod{ h_1}{ S_{\xi} h_2} + \int_0^t \sprod{ h_1 }{ V_{\xi + se_\beta} A_\beta(\xi + se_\beta) h_2 } ds \;
  . \label{eq:intrel}
\end{eqnarray}
By an approximation argument using that  $|p|$ is a closed operator
and that $\HH \cap ( C_0^\infty(\R^{3N}) \otimes \FF )$ is a core
for $|p|$, we conclude that
 \eqref{eq:intrel} holds  for all $h_2 \in D(|p|)$.
For $h_2 \in D(|p|)$, $\xi \mapsto S_{\xi} A_\beta(\xi) h_2$ is
continuous and therefore \eqref{eq:intrel}  holds in fact in the
strong sense, i.e.,
$$
V_{\xi + t e_\beta} h_2 = S_{\xi} h_2 + \int_0^t V_{\xi + se_\beta}
A_\beta(\xi + se_\beta) h_2 ds \quad , \quad \forall h_2 \in D(|p|).
$$
This implies that for all $h_2 \in D(|p|)$,  $t \mapsto V_{\xi +
te_\beta} h_2$ is $C^1$ with derivative $ \partial_\beta S_{\xi} h_2
= S_{\xi} A_\beta(\xi) h_2$.\\

\noindent \underline{Step 2}: Suppose $\xi \mapsto \psi(\xi)$ is a
$C^1$ function such that  $\psi(\xi) \in D(|p|)$ and $\xi \mapsto
A_\beta(\xi) \psi(\xi)$ is continuous for all $\beta$. Then $\xi
\mapsto S_{\xi} \psi(\xi)$ is in $C^1$ and for all $\beta$,
\begin{eqnarray} \label{app:part}
\partial_\beta S_{\xi} \psi(\xi) =  S_{\xi} A_\beta(\xi) \psi(\xi) + S_{\xi}
\partial_\beta \psi(\xi)   \; .
\end{eqnarray}

Using  the differentiability of $\xi \mapsto \psi(\xi)$, $\psi(\xi)
\in D(|p|)$, and Step 1,
we see by the product rule of differentiation that $\xi \mapsto
S_{\xi} \psi(\xi)$ is differentiable with partial derivative
\eqref{app:part}. \eqref{app:part}  depends
continuously on $\xi$, because $\xi \mapsto S_{\xi}$ is strongly
continuous and, by assumption, both, $\xi \mapsto
\partial_\beta \psi_\xi$  and $\xi \mapsto A_\beta(\xi)
\psi_\xi$ are continuous.\\

\noindent \underline{Step 3}: For $\xi$ in a
neighborhood of zero:
\begin{itemize}
\item[(a)] $\varphi_\xi \in D(p^2)$, and the functions
$ \xi \mapsto p_\mu \varphi_\xi $ and $ \xi \mapsto p_\mu p_\nu
\varphi_\xi $ are analytic for all $\mu,\nu$.
\item[(b)] For all $\beta$, $\partial_\beta \varphi_\xi \in D(|p|)$
and $\xi \mapsto p_\mu \partial_\beta \varphi_\xi$ is analytic for
all $\mu$.
\end{itemize}

(a) For $h$ from a dense subset of $\HH$,
$\langle h, p_\nu p_\mu\varphi_\xi \rangle=\langle p_\nu p_\mu h , \varphi_\xi \rangle$,
which is analytic in $\xi$. Since, by \eqref{eq:uniformbound} below, $\|p_{\mu}p_{\nu}\varphi_\xi\|$ is locally
bounded, the analyticity of $\xi\mapsto p_{\mu}p_{\nu}\varphi_\xi$ follows by an approximation argument (Remark
III-1.38 in \cite{kat:per}). To prove the bound \eqref{eq:uniformbound},
we  use that, for small $\xi$, $\widetilde{H}(\xi)$ is an analytic
family of type (A) and the Coulomb potential is
infinitesimally Laplacian bounded, i.e.,
\begin{align}
\| p^2 \varphi_\xi \| &\leq  \| ( p^2 + H_f ) \varphi_\xi \|\nonumber \\
&\leq  {\rm const.} ( \| \widetilde{H}(\xi) \varphi_{\xi} \| + \|
\varphi_\xi \| )  = {\rm const.} ( \| E(\xi) \varphi_\xi \| + \|
\varphi_\xi \| ) \; . \label{eq:uniformbound}
\end{align}
 The
analyticity of  $\xi \mapsto p_\mu  \varphi_\xi$ follows using the
same arguments as above and the bound $\|p_\mu  \varphi_\xi \|^2
\leq \|\varphi_\xi \| \|p^2 \varphi_\xi\|$.

(b)  Since for all $\mu$, the operator $p_\mu $ is  closed and $\xi
\mapsto p_\mu  \varphi_\xi$ is analytic, we have $\partial_\beta
\varphi_\xi \in D(|p|)$ and $\partial_\beta p_\mu  \varphi_\xi =
p_\mu
\partial_\beta \varphi_\xi$.\\

\noindent \underline{Step 4}: For all $\alpha$, the functions $\xi
\mapsto\varphi_\xi$, $\xi \mapsto
\partial_\alpha \varphi_\xi$,  and $\xi \mapsto A_\alpha(\xi) \varphi_\xi$  satisfy the assumptions of Step 2.
In particular, $ \xi \mapsto S_{\xi} \varphi_\xi$ is of class $C^2$.\\

An iteration of Step~2 shows that the statement of the first
sentence implies the statement of the second sentence.  To prove the former,
we recall that analytic functions are of class $C^{\infty}$ \cite{Krantz}, and we begin with
the following observation. If $\xi \mapsto\psi(\xi) \in D(|p|)$  and $\xi \mapsto p_\mu \psi(\xi)$ are  in
$C^1$ for all $\mu$, then $\xi \mapsto A_\beta(\xi) \psi(\xi)$ is in
$C^1$ for all $\beta$, which follows using expression
\eqref{eq:infgen}. If, moreover, $\psi(\xi) \in D(p^2)$ and $\xi
\mapsto p_\mu p_\nu \psi(\xi)$ is in $C^1$ for all $\mu,\nu$, then
$A_\alpha(\xi) \psi(\xi) \in D(|p|)$ and $\xi \mapsto A_\beta(\xi)
A_\alpha(\xi) \psi(\xi)$ is in $C^1$ for all $\alpha,\beta$.
Applying, these properties to $\varphi_\xi$ and using Step~3~(a), we
see that $\xi \mapsto \varphi_\xi$ and $\xi \mapsto A_\alpha(\xi)
\varphi_\xi$ satisfy the assumptions of Step 2. Similarly, using
Step~3~(b), we see that $\xi \mapsto\partial_\alpha \varphi_\xi$
satisfies the assumptions of Step~2.
\end{proof}


\section{The Smooth Feshbach Map}
\label{sec:fesh}

In this section we describe the smooth Feshbach
transform of Bach et al. \cite{bacchefrosig:smo} in a slightly generalized
form that allows for non self-adjoint smoothed projections. There are further small differences
between our presentation here and the one of \cite{bacchefrosig:smo}, which are
explained in \cite{grihas:smo}.

Let $\chi$ and $\overline{\chi}$ be commuting, nonzero bounded
operators, acting on a separable Hilbert space $\HH$ and
satisfying $\chi^2 + \overline{\chi}^2 = 1$. By a \emph{Feshbach
pair} $(H,T)$ for $\chi$ we mean a pair of closed operators with
same domain
$$
    H,T: D(H)=D(T)\subset \HH\to\HH
$$
such that $H,T,W:=H-T$, and the operators
\begin{align*}
W_\chi &:= \chi W \chi,& W_{\chib} &:= \chib W \chib, \\
H_\chi &:= T + W_\chi,& H_{\chib} &:= T +  W_{\chib},
\end{align*}
defined on $D(T)$ satisfy the following assumptions:
\begin{itemize}
\item[(a)] $\chi T\subset T\chi$ and $\chib T\subset T\chib$,
\item[(b)] $T, H_{\chib}:D(T)\cap\ran\chib\to \ran\chib$ are bijections with bounded inverse.
\item[(c)] $\chib H_{\chib}^{-1} \chib W \chi: D(T)\subset\HH\to\HH$ is a bounded operator.
\end{itemize}

Henceforth we will call an operator $A: D(A) \subset \HH \to \HH$
\emph{bounded invertible in a subspace} $Y \subset \HH$ ($Y$ not
necessarily closed), if $A : D(A) \cap Y \to Y$ is a bijection
with bounded inverse.

\noindent\emph{Remarks.}
\begin{enumerate}
\item To verify (a), it suffices to show that $T\chi=\chi T$ and $T\chib=\chib
T$ on a \emph{core} of $T$.

\item If $T$ is bounded invertible in $\ran\chib$, $\|T^{-1}\chib W\chib\|<1$, 
$\|\chib WT^{-1}\chib\|<1$ and $T^{-1}\chib W\chi$ is bounded, then the bounded
invertibility of $H_{\chib}$ and condition (c) follow. See
Lemma~\ref{lm:F-cond} below.

\item Note that $\ran\chi$ and $\ran\chib$ need not be closed and
are not closed in the application to QED. One can however, replace
$\ran\chib$ by $\overline{\ran\chib}$ both in condition (b) and in
the statement of Theorem~\ref{sfm}, below. Then this theorem continues to
hold and the proof remains unchanged.
\end{enumerate}

Given a Feshbach pair $(H,T)$ for $\chi$, the operator
\begin{equation}\label{eq:fesh}
    F_\chi(H,T) := H_{\chi} - \chi W \chib H_{\chib}^{-1} \chib W \chi
\end{equation}
on $D(T)$ is called \emph{Feshbach map of} $H$. The mapping
$(H,T)\mapsto F_\chi(H,T)$ is called \emph{Feshbach map}. The
auxiliary operators
\begin{eqnarray*}
\begin{array}{ll}
Q_\chi := \chi - \chib  H_{\chib}^{-1} \chib W \chi  \\
Q_\chi^\# := \chi - \chi W \chib H_{\chib}^{-1} \chib,
\end{array}
\end{eqnarray*}
play an important role in the analysis of $F_\chi(H,T)$. By
conditions (a), (c), and the explanation above, they are bounded, and
$Q_{\chi}$ leaves $D(T)$ invariant. The Feshbach map is
isospectral in the sense of the following Theorem, which generalizes
Theorem 2.1 in \cite{bacchefrosig:smo} to non-selfadjoint $\chi$ and
$\chib$.

\begin{theorem}\label{sfm} Let $(H,T)$ be
a Feshbach pair for $\chi$ on a separable Hilbert space $\HH$.
Then the following holds:
\begin{enumerate}
\item[(i)] Let $Y$ be a subspace with  $\ran \chi \subset Y
\subset \HH$,
\begin{equation} \label{eq:feshsubspace} T: D(T) \cap Y \to Y,\qquad{\rm  and}\qquad\chib T^{-1} \chib Y\subset Y \; .
\end{equation}
Then $H:D(H)\subset\HH \to \HH$ is bounded invertible if and only
if $F_\chi(H,T): D(T) \cap Y \to Y$ is bounded invertible in $Y$.
Moreover,
\begin{eqnarray*}
   H^{-1} &=& Q_{\chi}F_{\chi}(H,T)^{-1} Q_{\chi}^{\#} + \chib H_{\chib}^{-1}\chib,\\
    F_\chi(H,T)^{-1} &=& \chi H^{-1}\chi + \chib T^{-1}\chib.
\end{eqnarray*}

\item[(ii)] $\chi \ker H
\subset \ker F_{\chi}(H,T)$ and $Q_{\chi}\ker F_{\chi}(H,T)\subset\ker H$.
The mappings
\begin{align}
\chi :&  \  \ker H \to \ker F_\chi(H,T)  \label{eq:keriso1}, \\
Q_\chi :&  \ \ker F_\chi(H,T) \to \ker H  \label{eq:keriso2},
\end{align}
are linear isomorphisms and inverse to each other.
\end{enumerate}
\end{theorem}

\noindent\emph{Remarks.}
\begin{enumerate}
\item The subspaces $Y = \ran \chi$ and $Y = \HH$ satisfy the
conditions stated in \eqref{eq:feshsubspace}.
\item From \cite{bacchefrosig:smo} it is known that $\chi$ and $Q_{\chi}$
are one-to-one on $\ker H$ and $\ker F_{\chi}(H,T)$ respectively. The
stronger result (ii) will be derived from the new algebraic
identities (a) and (b) of the following lemma.
\end{enumerate}

Theorem~\ref{sfm} easily follows from the next lemma, which is of
interest and importance in its own right.

\begin{lemma}\label{lm:F-basics}
Let $(H,T)$ be a Feshbach pair for $\chi$ and let
$F:=F_{\chi}(H,T)$, $Q := Q_\chi$, and $Q^{\#} := Q^{\#}_\chi$ for
simplicity. Then the following identities hold:
\begin{align*}
(a)&& (\chib H_{\chib}^{-1} \chib)H &= 1-Q\chi,\quad\text{on}\
D(T),&\qquad
H(\chib H_{\chib}^{-1} \chib) &= 1-\chi Q^{\#},\quad \text{on}\ \HH,\\
(b)&& (\chib T^{-1}\chib)F &= 1-\chi Q,\quad\text{on}\
D(T),&\qquad
F(\chib T^{-1}\chib)&=1-Q^{\#}\chi,\quad \text{on}\ \HH,\\
(c)&& HQ &=\chi F,\quad\text{on}\ D(T),&\qquad Q^{\#}H &= F\chi,
\quad\text{on}\ D(T).\\
\end{align*}
\end{lemma}

For the proofs of Lemma~\ref{lm:F-basics} and Theorem~\ref{sfm} we refer to
\cite{grihas:smo}.

\begin{lemma}\label{lm:F-cond}
Conditions (a),(b), and (c) on Feshbach pairs are satisfied if
\begin{itemize}
\item[(a')] $\chi T\subset T\chi$ and $\chib T\subset T\chib$,
\item[(b')] $T$ is bounded invertible in $\ran\chib$,
\item[(c')] $\|T^{-1}\chib W\chib\|<1$, $\|\chib WT^{-1}\chib\|<1$ and
  $T^{-1}\chib W\chi$ is a bounded operator.
\end{itemize}
\end{lemma}

\begin{proof}
By assumptions (a') and (b'), on $D(T)\cap\ran\chib$,
$$
    H_{\chib} = (1+\chib WT^{-1}\chib)T,
$$
and $T:D(T)\cap\ran\chib\to\ran\chib$ is a bijection with bounded
inverse. From (c') it follows that
$$
   1+\chib WT^{-1}\chib : \ran\chib \to\ran\chib
$$
is a bijection with bounded inverse. In fact, $(1+\chib W
T^{-1}\chib)\ran\chib\subset\ran\chib$, the Neumann series
$$
     \sum_{n\geq 0} (-\chib WT^{-1}\chib)^n = 1-\chib W T^{-1}\chib\sum_{n\geq 0}
     (-\chib WT^{-1}\chib)^n
$$
converges and maps $\ran\chib$ to $\ran\chib$. Hence
$H_{\chib}\upharpoonright\ran\chib$ is bounded invertible.

Finally, from $H_{\chib}=T(1+T^{-1}W_{\chib})$ and (c') it
follows that
$$
    H_{\chib}^{-1}\chib W\chi = (1+T^{-1}W_{\chib})^{-1}T^{-1}\chib W\chi,
$$
which, by (c'), is bounded.
\end{proof}


\section{The Initial Hamiltonian on Fock Space}
\label{sec:effect}

As explained in the introduction, the first step in our
renormalization analysis of $H_{g}(s)$ is to use the Feshbach map
to define an isospectral operator $H^{(0)}[s,z]$ on
$$
 \HH_{\rm red} := P_{[0,1]}(H_f)\FF.
$$
Let $\chi, \chib \in C^\infty(\R ; [0,1] )$ with
$\chi(t) = 1$ if $t\leq 3/4$,  $\chi(t)=0$ if $t \geq 1$ and
$\chi^2 + \chib^2 = 1$. For $\rho>0$ we define operators $\chi_{\rho} := \chi(H_f/\rho),\
\chib_{\rho} := \chib(H_f/\rho)$, and
\begin{eqnarray*}
\boldsymbol{\chi}(s) &:=& P_{\rm at}(s) \otimes \chi_1, \\
\overline{\boldsymbol{\chi}}(s) &:=& \overline{P}_{\rm at}(s)
\otimes 1 + P_{\rm at}(s) \otimes \overline{\chi}_1 \; .
\end{eqnarray*}
By assumption on $\chi$ and $\chib$,
$$
\boldsymbol{\chi}(s)^2  + \overline{\boldsymbol{\chi}}(s)^2=1,
$$
but $\boldsymbol{\chi}(s)$ and
$\overline{\boldsymbol{\chi}}(s)$ will not be self-adjoint unless $s$ is real.
Since $ P_{\rm at}(s)$ is a bounded projection with one-dimensional range, any
linear operator $L$ in $\HH_{\mathrm at}\otimes\FF$ that is defined and bounded on
$\ran P_{\rm at}(s) \otimes \FF$, defines a unique bounded linear transformation $\langle L
\rangle_{{\rm at},s}$ on $\FF$, through the equation
\begin{equation} \label{eq:eff2}
( P_{\rm at}(s) \otimes  1 ) L ( P_{\rm at}(s) \otimes 1 ) =
P_{\rm at}(s) \otimes \langle L \rangle_{{\rm at},s}.
\end{equation}

We are no ready to define the effective Hamiltonian $H^{(0)}[s,z]$ on $\HH_{\rm red}$. To this end we
assume, for the moment, that $(H_g(s)-z, H_0(s) - z)$ is a
Feshbach pair for $\boldsymbol{\chi}(s)$. This assumption will be
justified by Theorem \ref{thm:fesh} below. From $1 = P_{\rm at}(s)
+ \overline{P}_{\rm at}(s)$ and the fact that $P_{\rm at}(s)$ is a
rank one operator, we find
\begin{equation}\label{Heff1}
  F_{\boldsymbol{\chi}(s)}(H_g(s)-z, H_0(s) - z) = (H_0(s) - z )
  \overline{P}_{\rm  at}(s) \otimes 1 + P_{\rm at}(s) \otimes
  \tilde{H}^{(0)}[s,z] \; ,
\end{equation}
with
\begin{equation} \label{eq:defh0}
\tilde{H}^{(0)}[s,z] =  E_{\rm at}(s) - z + H_f + W_{\rm at}[s,z]
\end{equation}
and $W_{\rm at}[s,z] \in \mathcal{L}(\mathcal{F})$ given by
\begin{eqnarray} \label{eq:eff22}
  W_{\rm at}[s,z]
 & =& g \langle  \chi_1   W(s)  \chi_1 \rangle_{{\rm at},s}  \\
 \nonumber
 && - g^2  \langle  \chi_1 W(s) \overline{\boldsymbol{\chi}}(s) (H_g(s) -
z)_{\overline{\boldsymbol{\chi}}(s)}^{-1}
\overline{\boldsymbol{\chi}}(s) W(s)  \chi_1 \rangle_{{\rm at},s}\;
.
\end{eqnarray}
The operators $\tilde{H}^{(0)}[s,z]$ and $H_g(s)-z$ are isospectral
in the sense of Theorem~\ref{sfm}. More explicitly, the following
proposition holds true.

\begin{prop}\text{} \label{pro:initial} Let
 $(H_g(s)-z, H_0(s) - z)$ be a Feshbach pair for $\boldsymbol{\chi}(s)$.
 Then:
\begin{itemize}
\item[(i)]  $H_g(s)-z:D(H_0(s))\subset\HH \to \HH$ is bounded
invertible if and only if
  $\tilde{H}^{(0)}[s,z]$ is bounded invertible on  $\HH_{\rm red}$.
\item[(ii)] The following maps are linear isomorphisms and
inverses of each other:
\begin{eqnarray*}
\boldsymbol{\chi}(s)  : \ \ker (H_g(s) -
z) & \longrightarrow & P_{\rm at}(s) \HH_{\rm at} \otimes \ker \tilde{H}^{(0)}[s,z], \\
Q_{\boldsymbol{\chi}(s)} :  \ P_{\rm at}(s) \HH_{\rm at} \otimes
\ker \tilde{H}^{(0)}[s,z] & \longrightarrow & \ker (H_g(s) - z).
\end{eqnarray*}
\end{itemize}
\end{prop}

\begin{proof}
(i) We fix $(s,z)$ and for notational simplicity suppress the $s$
and $z$ dependence. Let  $Y = \ran(P_{\rm at} \otimes
P_{[0,1]}(H_f))$.
Then $H_0 :
D(H_0) \cap Y \to Y$ and $ \boldsymbol{\chib} (H_0 - z )^{-1}
\boldsymbol{\chib} Y \subset Y$. By Theorem \ref{sfm},
\begin{eqnarray*}
\lefteqn{H_g - z \ {\rm is \ bounded \ invertible \ in}  \ \HH} \\
 &\Leftrightarrow & F_{\boldsymbol{\chi}}(H_g-z,H_0-z) \ {\rm is \ bounded \
invertible \ in}\ Y\\
& \Leftrightarrow & \tilde{H}^{(0)} \ {\rm is \ bounded \ invertible
\ on } \ \HH_{\rm red},
\end{eqnarray*}
where the last equivalence follows from $F_{\boldsymbol{\chi}}(H_g-z,H_0-z)=1\otimes \tilde{H}^{(0)}$ on
$Y$.\\
Statement (ii) follows immediately from Theorem~\ref{sfm}, (ii).
\end{proof}

Since $\ran\boldsymbol{\chi}(s)\subset P_{\rm at}(s)\HH_{\rm at}\otimes\HH_{\rm
  red}$, Proposition~\ref{pro:initial} (ii) implies that
$$
\ker \tilde{H}^{(0)}[s,z] = \ker \tilde{H}^{(0)}[s,z] \cap \HH_{\rm
red} = \ker (\tilde{H}^{(0)}[s,z] \upharpoonright \HH_{\rm red}) \; .
$$
Therefore, and because of Proposition~\ref{pro:initial} (i), it is sufficient for
our purpose to study the restriction
$$
H^{(0)}[s,z] := \tilde{H}^{(0)}[s,z] \upharpoonright \HH_{\rm red}.
$$
In the remainder of this section we use Hypotheses~\ref{hyp:G}-\ref{hyp:R} to
verify, for small $g$, the assumption of Proposition~\ref{pro:initial} and to
show that $H^{(0)}[s,z]$ is analytic on $\UU$. To this end we need the
following lemmas.


\begin{lemma} \label{lem:eight}
Suppose that Hypotheses \ref{hyp:G}--\ref{hyp:R} hold. Then,
for all $(s,z) \in \UU$, $H_0(s)-z$ is bounded invertible on
$\ran\obchi(s)$ and
\begin{eqnarray} \label{eq:344}
\sup_{(s,z)\in\UU} \left\|(H_f+1)(H_0(s)-z)^{-1}\obchi(s)\right\|
&<& \infty \\
\nonumber\sup_{(s,z) \in \UU}
\left\|W(s)(H_0(s)-z)^{-1}\obchi(s)\right\| &<& \infty\\
\nonumber\sup_{(s,z) \in \UU} \left\|(H_0(s) - z )^{-1}\obchi(s)
W(s) \right\| &<& \infty .
\end{eqnarray}
\end{lemma}

\begin{proof}
The parameter $s$ is suppressed in this proof to make long
expressions more readable. Recall that $\obchi H_0\subset
H_0\obchi$. Hence $H_0-z$ maps $D(H_0)\cap\ran\obchi$ into
$\ran\obchi$. Moreover,
$$
    \ran \obchi = \ran(\Patb\otimes 1)\oplus\ran(\Pat\otimes \chib_1)
$$
where
\begin{eqnarray}\label{eq:eight1}
    H_0 - z &:& D(H_0)\cap \ran(\overline{P}_{\rm at}\otimes
    1)\to \ran(\overline{P}_{\rm at}\otimes 1),\\
    H_0 - z &:& D(H_0)\cap \ran(P_{\rm at}\otimes
    \chib_1)\to \ran(P_{\rm at}\otimes \chib_1).\label{eq:eight2}
\end{eqnarray}
Working in a spectral representation where $H_f$ is multiplication
by $q\geq 0$, it is easily seen from Hypothesis~\ref{hyp:R} that
\eqref{eq:eight1} and \eqref{eq:eight2} are bounded invertible for $(s,z)\in\UU$,
and hence that $(H_0-z):D(H_0)\cap\ran\obchi\to \ran\obchi$ is a
bijection. The inverses of \eqref{eq:eight1} and \eqref{eq:eight2}
are bounded by
\begin{eqnarray}\label{eq:eight3}
   \|(H_0-z)^{-1}\Patb\otimes 1\| &\leq
   &\sup_{(s,z)\in\UU}\sup_{q\geq 0}\|(\Hat-z+q)^{-1}\Patb\|,\\
\|(H_0-z)^{-1}\Pat\otimes \chi(H_f\geq 3/4)\| &\leq
&\sup_{(s,z)\in\UU}\sup_{q\geq
3/4}\left|\frac{1}{\Eat-z+q}\right|\|\Pat\|.\label{eq:eight4}
\end{eqnarray}
Since $\obchi = [\Patb\otimes 1]\obchi+ [\Pat\otimes\chi(H_f\geq
3/4)]\obchi$ it follows from \eqref{eq:eight3}, \eqref{eq:eight4}
and Hypothesis~\ref{hyp:R} that
$$
  \sup_{(s,z)\in\UU}\|(H_0-z)^{-1}\upharpoonright\ran\obchi\|<\infty.
$$

Bound \eqref{eq:344} is proved in a similar way, using
\begin{eqnarray}\label{eq:eight5}
   \|(H_f+1)(H_0-z)^{-1}\Patb\otimes 1\| &\leq
   &\sup_{(s,z)\in\UU}\sup_{q\geq 0}\left\|\frac{q+1}{\Hat-z+q}\Patb\right\|,\\
\|(H_f+1)(H_0-z)^{-1}\Pat\otimes \chi(H_f\geq 3/4)\| &\leq
&\sup_{(s,z)\in\UU}\sup_{q\geq
3/4}\left|\frac{q+1}{\Eat-z+q}\right|\|\Pat\|,\label{eq:eight6}
\end{eqnarray}
instead of \eqref{eq:eight3} and \eqref{eq:eight4}. The right
sides of \eqref{eq:eight5} and \eqref{eq:eight6} are finite by
Hypothesis~\ref{hyp:R}.

The remaining inequalities of Lemma~\ref{lem:eight} follow from
\eqref{eq:344} and
\begin{eqnarray*}
\sup_{s}\|W(s) (H_f+1)^{-1}\| &\leq & \sup_{s}\|G_s\|_{\omega},\\
\sup_{s}\|(H_f+1)^{-1}W(s)\| &\leq & \sup_{s}\|G_s\|_{\omega},
\end{eqnarray*}
where $\sup_{s}\|G_s\|_{\omega}<\infty$ by Hypothesis~\ref{hyp:G}.
\end{proof}


\begin{lemma}\label{lem:analyticgs}
The mapping $s \mapsto W(s)(H_f + 1 )^{-1/2} \in \mathcal{L}(\HH)$ is
analytic  on $V$.
\end{lemma}

\begin{proof}  
From
$$
\| W(s) (H_f + 1 )^{-1/2} \| \leq 2 \| G_s \|_\omega
$$
we see, by Hypothesis I,  that $s \mapsto W(s) (H_f + 1 )^{-1/2}$ is
uniformly bounded. By this uniform bound (see Theorem III-3.12 of
\cite{kat:per}) it is sufficient to show that the function
\begin{equation} \label{eq:weakanalyticity}
s \mapsto (\psi_1, W(s) (H_f + 1 )^{-1/2} \psi_2) \; ,
\end{equation}
is analytic on $V$, for all $\psi_1, \psi_2$ in the dense linear
subspace spanned by all vectors of the form $\varphi \otimes S_n (h_1 \otimes h_2\otimes \cdots \otimes h_n)$,
with $\varphi \in \HH_{\rm at}$ and $h_i \in \mathfrak{h}$,
$n\in\N$. For such vectors, \eqref{eq:weakanalyticity} is a linear
combination of terms of the form $(\varphi_1 \otimes h, G_s
\varphi_2)$, with $\varphi_1, \varphi_2 \in \HH_{\rm at}$ and $h \in
\mathfrak{h}$. They are analytic by Hypothesis I.
\end{proof}


\begin{theorem}  \label{thm:fesh}
Suppose Hypotheses \ref{hyp:G}--\ref{hyp:R} hold, and let $\UU\subset
V\times\C$ be given by Hypothesis \ref{hyp:R}. Then there exists a $g_0>0$ such that
for all $(s,z)\in\UU$ and for all $g\in [0,g_0)$, the pair $(H_g(s)-z, H_0(s) - z)$ is a Feshbach pair for $\boldsymbol{\chi}(s)$.
Moreover, $H_g^{(0)}[s,z]$ is analytic on $\UU$.
\end{theorem}

\begin{proof}
To prove that $H_g(s)=H_0(s)+gW(s)$ is closed on $D(H_0(s))$ for all
$g\in\R$, we prove that $W(s)$ is infinitesimally bounded with
respect to $H_0(s)$. Suppose that $(s,z)\in \UU$ for some $z \in
\C$. By Hypothesis~\ref{hyp:G},
\begin{equation}\label{eq:fesh1}
    \|W(s)(H_f+1)^{-1/2}\| \leq 2\|G_s\|_{\omega} <\infty.
\end{equation}
On the other hand, by the reasoning in the proof of
\eqref{eq:344}, Hypothesis \ref{hyp:R} implies that $z-q \in
\rho(H_{\rm at}(s))$ for $q \geq 1$, that $w:=z-1 \in
\rho(H_0(s))$, and that
\begin{equation}\label{eq:fesh2}
  H_f (H_0(s) - w)^{-1} \in \LL(\HH).
\end{equation}
Combining \eqref{eq:fesh1} and \eqref{eq:fesh2} we see that, for
all $\ph\in D(\Hat(s))\otimes D(H_f)$,
\begin{eqnarray*}
  \|W(s)\ph\|^2 &\leq & C_0 \sprod{\ph}{(H_f+1)\ph}\\
  &=& C_0 \sprod{\ph}{(H_f+1)(H_{0}(s)-w)^{-1}(H_{0}(s)-w)\ph}\\
  &\leq & C_1 \|\ph\| \|H_0(s)\ph\| + C_2 \|\ph\|^2\\
  & \leq & C_1 \eps \|H_0(s)\ph\|^2 + \left(\frac{C_1}{\eps}+C_2\right)\|\ph\|^2
\end{eqnarray*}
with constants $C_0,C_1,C_2$.

Next we verify the criteria for Feshbach pairs from
Lemma~\ref{lm:F-cond}. Obviously, $\bchi(s) H_0(s) = H_0(s)\bchi(s)$ and
$\obchi(s) H_0(s) = H_0(s)\obchi(s)$ on $D(\Hat)\otimes D(H_f)$. By
the first remark of Section~\ref{sec:fesh}, this proves condition (a') of
Lemma~\ref{lm:F-cond}.

By Lemma \ref{lem:eight}, $H_0(s)-z$ is bounded invertible on
$\ran\obchi(s)$ and
\begin{align}\label{eq:fesh3}
&\sup_{(s,z) \in U} \left\| g \obchi(s)W(s)(H_0(s)-z)^{-1}\obchi(s)\right\|< 1,\\
& \sup_{(s,z) \in U} \left\|( H_0(s) - z)^{-1}\obchi(s)g
W(s)\obchi(s)\right\|
 < 1,\nonumber
\end{align}
for  $g$ sufficiently small. This proves (b') and (c') of
Lemma~\ref{lm:F-cond} and hence completes the proof that
$(H_g(s)-z, H_0(s) - z)$ is a Feshbach pair.


It remains to prove the analyticity of $H^{(0)}[s,z]\restricted
\HH_{\rm red}$. By \eqref{Heff1},  $H^{(0)}(s,z)$ is analytic if
$W_{\rm at}[s,z]$ is analytic. We will show that
\begin{eqnarray} \label{eq:75}
(s,z) \mapsto \chi_1\left( g W(s) -  g^2 W(s)
\overline{\boldsymbol{\chi}}(s) (H_g(s) - z)_{\obchi(s)}^{-1}
\obchi(s) W(s) \right) \chi_1
\end{eqnarray}
is analytic in $s$ and $z$. By Eqns.~\eqref{eq:eff2} and
\eqref{eq:eff22} this will imply the analyticity of $\langle
\alpha , W_{\rm at}[s,z] \beta \rangle$ for all $\alpha, \beta \in
\FF$, which, by Theorem 3.12 of Chapter III in \cite{kat:per},
proves that $W_{\rm at}[s,z]$ is analytic in $s$ and $z$.

Since $\chi_1 W(s)$ and $W(s)\chi_1$ are analytic the analyticity
of \eqref{eq:75} follows if we show that
\begin{equation}\label{eq:fesh4}
    (s,z) \mapsto (H_g(s) - z)_{\obchi(s)}^{-1} \obchi(s)
\end{equation}
is analytic. Assuming that $|g|$ is small enough for
\eqref{eq:fesh3} to hold, the Neumann series
\begin{eqnarray}\label{eq:fesh5}
(H_g - z)_{\obchi}^{-1} \vert_{\ran \obchi} = (H_0 - z )^{-1}
\sum_{n=0}^\infty \left(-\obchi g W ( H_0 - z)^{-1}\obchi\right)^n
\big\vert_{\ran \obchi}
\end{eqnarray}
converges uniformly for $(s,z)\in \UU$. Hence \eqref{eq:fesh4}
will be analytic if each term of the series \eqref{eq:fesh5} is
analytic. By Lemma~\ref{lem:analyticgs}, $W(s)(H_f + 1)^{-1}$ is
analytic. Hence it remains to prove analyticity of
$$
( H_f + 1 )(H_0(s) - z )^{-1}\obchi(s)\vert_{\ran \obchi(s)}.
$$
By the definition of $\obchi(s)$,
\begin{eqnarray*}
\lefteqn{ (H_f + 1)(H_0(s) - z )^{-1} \obchi(s)} \\
&& = (H_f + 1 )(H_0(s) - z)^{-1} ( \overline{P}_{\rm at}(s) \otimes
1 ) + (H_f + 1 )(E_{\rm at}(s) + H_f - z )^{-1}(P_{\rm at}(s)
\otimes \overline{\chi}_1 )\; .
\end{eqnarray*}
The factor $(H_f + 1 )(E_{\rm at}(s) + H_f  - z )^{-1}$ in the
second term on the r.h.s. can be viewed as a composition of
analytic functions. The analyticity of the first term on the
r.h.s. is derived, in a spectral representation of $H_f$, from
Hypothesis \ref{hyp:R}, and Proposition \ref{pro:analytredres} of
the Appendix.
\end{proof}


\section{The  Renormalization Transformation}
\label{sec:renorm}

The renormalization transformation is defined on a subset
of $\LL(\HH_{\rm red})$ that will be parameterized by
vectors of a Banach space $\WW_{\xi}=\oplus_{m,n\geq 0}\WW_{m,n}$.
We begin with the definition of this Banach space.

The Banach space $\WW_{0,0}$ is the space of continuously differentiable functions
\begin{eqnarray*}
   \WW_{0,0} &:=& C^1([0,1])\\
   \|w\| &:=& \|w\|_{\infty}+\|w'\|_{\infty}
\end{eqnarray*}
where $w'(r):=\partial_r w(r)$. For $m,n\in\N$
with $m+n\geq 1$ and $\mu>0$ we set
\begin{eqnarray*}
   \WW_{m,n} &:=&
   L^2_s\left(B^{m+n},\frac{dK}{|K|^{2+2\mu}};\WW_{0,0}\right)\\
    \|w_{m,n}\|_{\mu} &:=& \left(\int_{B^{m+n}}
    \|w_{m,n}(K)\|^2\frac{dK}{|K|^{2+2\mu}}\right)^{1/2}
\end{eqnarray*}
where $B:=\{k\in\R^3\times\{1,2\}:|k|\leq 1\}$ and
$$
     |K|:=\prod_{j=1}^{m+n}|k_j|,\qquad dK:=\prod_{j=1}^{m+n}dk_j.
$$
That is, $\WW_{m,n}$ is the space of measurable functions $w_{m,n}:B^{m+n}\to
\WW_{0,0}$ that are symmetric with respect to all permutations of
the $m$ arguments from $B^{m}$ and the $n$ arguments from $B^{n}$,
respectively, such that $\|w_{m,n}\|_{\mu}$ is finite.

For given $\xi\in (0,1)$ and $\mu>0$ we define a Banach space
\begin{eqnarray*}
   \WW_{\xi} &:=& \bigoplus_{m,n\in \N} \WW_{m,n}\\
    \|w\|_{\mu,\xi} &:=& \sum_{m,n\geq 0}\xi^{-(m+n)}\|w_{m,n}\|_{\mu},
\end{eqnarray*}
$\|w_{0,0}\|_{\mu}:=\|w_{0,0}\|$, as the completion of the linear space of finite sequences
$w=(w_{m,n})_{m,n\in\N}\in \bigoplus_{m,n\in \N} \WW_{m,n}$ with respect to
the norm $\|w\|_{\mu,\xi}$.
The spaces $\WW_{m,n}$ will often be identified with the corresponding
subspaces of $\WW_{\xi}$.

Next we define a linear mapping $H:\WW_{\xi}\to\LL(\HH_{\rm
red})$. For \emph{finite} sequences $w=(w_{m,n})\in \WW_{\xi}$ the
operator $H(w)$ is the sum
$$
  H(w) := \sum_{m,n}H_{m,n}(w)
$$
of operators $H_{m,n}(w)$ on $\HH_{\rm red}$, defined by
$H_{0,0}(w) := w_{0,0}(H_f),$ and, for $m+n\geq 1$,
$$
  H_{m,n}(w) := P_{\rm red}\left(\int_{B^{m+n}}
  a^{*}(k^{(m)})w_{m,n}(H_f,K) a(\tilde{k}^{(n)}){dK}\right) P_{\rm
  red},
$$
where $P_{\rm red}:=P_{[0,1]}(H_f)$, $K=(k^{(m)},\tilde{k}^{(n)})$, and
\begin{align*}
   k^{(m)} &= (k_1,\ldots,k_m)\in (\R^3\times\{1,2\})^{m},& a^{*}(k^{(m)}) &= \prod_{i=1}^{m}a^{*}(k_i),\\
   \tilde{k}^{(n)}&=(\tilde{k}_1,\ldots,\tilde{k}_n)\in (\R\times\{1,2\})^{n},& a(\tilde{k}^{(n)})&=\prod_{i=1}^{n}a(\tilde{k}_i).
\end{align*}
By the continuity established in the following proposition, the
mapping $w\mapsto H(w)$ has a unique extension to a bounded linear
transformation on $\WW_{\xi}$.

\begin{prop}[\cite{bacchefrosig:smo}]\label{H-is-bounded}
(i) For all $\mu>0$, $m,n\in \N$, with $m+n\geq 1$, and $w\in
\WW_{m,n}$,
$$
   \|H_{m,n}(w) \| \leq \| ( P^\perp_\Omega H_f )^{-m/2}
   H(w_{m,n}) ( P^\perp_\Omega H_f )^{-n/2} \|
   \leq\frac{1}{\sqrt{m^m n^n}} \|w_{m,n}\|_\mu.
$$
(ii) For all $\mu>0$ and all $w\in \WW_{\xi}$
\begin{eqnarray}
   \|H(w)\|  &\leq &  \|w\|_{\mu,\xi} \nonumber \\
   \|H(w)\|  &\leq &  \xi\|w\|_{\mu,\xi},\qquad \text{if}\
   w_{0,0}=0. \label{eq:westimate}
\end{eqnarray}
In particular, the mapping $w\mapsto H(w)$ is continuous.
\end{prop}

\begin{proof}
Statement (ii) follows immediately from (i) and $\xi\leq 1$. For
(i) we refer to \cite{bacchefrosig:smo}, Theorem~3.1.
\end{proof}

Given $\alpha,\beta,\gamma\in \R_{+}$ we define neighborhoods,
$\BB(\alpha,\beta,\gamma)\subset H(\WW_{\xi})$ of the operator
$P_{\rm red}H_f P_{\rm red}\in \LL(\HH_{\rm red})$ by
$$
  \BB(\alpha,\beta,\gamma) := \big\{H(w)\big| |w_{0,0}(0)|\leq \alpha,\ \|w_{0,0}'-1\|_{\infty}\leq\beta,
  \ \|w-w_{0,0}\|_{\mu,\xi}\leq \gamma\big\}.
$$
Note that $w_{0,0}(0)=\sprod{\Omega}{w_{0,0}(H_f)\Omega}=\sprod{\Omega}{H(w)\Omega}$.
The definition of $\BB(\alpha,\beta,\gamma)$ is motivated by the following Lemma and by Theorem~\ref{bcfssigal}.

\begin{lemma} \label{feshbachtest} Suppose $\rho,\xi\in(0,1)$ and $\mu > 0$.
If $H(w) \in \BB(\rho/2, \rho/8 , \rho/8)$, then $(H(w),
H_{0,0}(w))$ is a Feshbach pair for $\chi_\rho$.
\end{lemma}

\begin{proof}
The assumption $H(w) \in \BB(\rho/2, \rho/8 , \rho/8)$ implies, by
Proposition \ref{H-is-bounded}, that
\begin{eqnarray*}
\| H(w) - H_{0,0}(w) \| \leq \xi \frac{\rho}{8} \; .
\end{eqnarray*}
For $r \in [\frac{3}{4} \rho , 1 ]$,
\begin{eqnarray*}
| w_{0,0}(r) | &\geq& r - | ( w_{0,0}(r) - w_{0,0}(0)) - r | - |w_{0,0}(0)| \\
&\geq & r ( 1 - \sup_r| {w'}_{0,0}(r) - 1 |) - \frac{\rho}{2}  \\
& \geq & \frac{3 \rho}{4} ( 1 - \frac{\rho}{8} ) - \frac{\rho}{2}
 \geq \frac{\rho}{8} \; .
\end{eqnarray*}
By the spectral Theorem,
\begin{eqnarray*}
\| H_{0,0}(w)^{-1} \restricted \ran \chib_\rho \| = \| w_{0,0}(H_f
)^{-1} \restricted \ran \chib_\rho \| \leq  \sup_{r \in
[\frac{3}{4} \rho, 1 ]} \frac{1}{|w_{0,0}(r) |}
 \leq  \frac{8}{\rho}  \; .
\end{eqnarray*}
Since $\| \chib_\rho \| \leq 1 $, it follows from the estimates
above that
$$
\| H_{0,0}(w)^{-1} \chib_\rho (H(w) - H_{0,0}(w)  ) \chib_\rho
\restricted \ran \chib_\rho \| \leq \xi < 1 \; .
$$
This implies the  bounded invertibility of
\begin{eqnarray*}
\lefteqn{ \left( H_{0,0}(w) + \chib_\rho (H(w) - H_{0,0}(w) )
\chib_\rho
\right) \restricted \ran \chib_\rho } \\
&& = H_{0,0}(w) \left( 1 + H_{0,0}(w)^{-1} \chib_\rho (H(w) -
H_{0,0}(w) ) \chib_\rho \right) \restricted \ran \chib_\rho \; .
\end{eqnarray*}
The other conditions on a Feshbach pair are now also satisfied,
since $H(w) - H_{0,0}(w)$ is bounded on $\HH_{\rm red}$.
\end{proof}


The \emph{renormalization transformation} we use is a composition
of a Feshbach transformation and a unitary scaling that puts the
operator back on the original Hilbert space $\HH_{\rm red}$.
Unlike the renormalization transformation of Bach et al
\cite{bacchefrosig:smo}, there is no analytic transformation of
the spectral parameter.

Given $\rho\in (0,1)$, let $\HH_{\rho}=\ran\chi(H_f\leq \rho)$. Let $w\in \WW_{\xi}$
and suppose $(H(w),H_{0,0}(w))$ is a Feshbach pair for $\chi_{\rho}$.
Then
$$
    F_{\chi_{\rho}}(H(w),H_{0,0}(w)): \HH_{\rho}\to \HH_{\rho}
$$
is iso-spectral with $H(w)$ in the sense of Theorem~\ref{sfm}. In
order to get a isospectral operator on $\HH_{\rm red}$, rather
than $\HH_{\rho}$, we use the linear isomorphism
$$
    \Gamma_{\rho}: \HH_{\rho}\to \HH_1=\HH_{\rm red},
    \qquad
    \Gamma_{\rho}:=\Gamma(U_{\rho})\upharpoonright\HH_{\rho},
$$
where $U_{\rho}\in \LL(L^2(\R^3\times\{1,2\}))$ is defined by
$$
    (U_{\rho}f)(k) := \rho^{3/2} f(\rho k).
$$
Note that $\Gamma_{\rho}H_f \Gamma_{\rho}^{*}=\rho H_f$, and hence
$\Gamma_{\rho}\chi_{\rho} \Gamma_{\rho}^{*}=\chi_1$. The
renormalization transformation $\RR_{\rho}$ maps bounded operators
on $\HH_{\rm red}$ to bounded linear operators on $\HH_{\rm red}$
and is defined on those operators $H(w)$ for which $(H(w),
H_{0,0}(w))$ is a Feshbach pair with respect to $\chi_{\rho}$.
Explicitly,
$$
   \RR_{\rho}(H(w)) :=
   \rho^{-1}\Gamma_{\rho}\FF_{\chi_{\rho}}(H(w),H_{0,0}(w))\Gamma_{\rho}^{*},
$$
which is a bounded linear operator on $\HH_{\rm red}$. In \cite{bacchefrosig:smo},
Theorem 3.3, it is shown that $w\mapsto H(w)$ is one-to-one. Hence $w\in
\WW_{\xi}$ is uniquely determined by the operator $H(w)$ and the
domain of $\RR_{\rho}$, as described above, is a well-defined
subset of $\LL(\HH_{\rm red})$. By Lemma~\ref{feshbachtest} it
contains the ball $\BB(\rho/2,\rho/8,\rho/8)$.

The following theorem describes conditions under which the Renormalization
transform may be iterated.

\begin{theorem}[\bf BCFS \cite{bacchefrosig:smo}] \label{bcfssigal}
There exists a constant $C_\chi\geq 1$  depending only on $\chi$,
such that the following holds. If $\mu > 0$, $\rho\in(0,1)$, $\xi
= \sqrt{\rho}/(4C_\chi)$, and $\beta, \gamma \leq \rho/(8C_\chi)$,
then
$$
\RR_\rho - \rho^{-1}  \langle \, \cdot \, \rangle_\Omega  :
\BB(\rho/2, \beta, \gamma) \to \BB(\alpha' , \beta' , \gamma' )\; ,
$$
where
$$
\alpha' = C_\beta \frac{\gamma^2}{\rho}
 \; , \quad \beta' = \beta +   C_\beta \frac{\gamma^2}{\rho}  \; ,
 \quad
  \gamma' = C_\gamma  \rho^\mu \gamma  \; ,
$$
with $C_\beta := \frac{3}{2} C_\chi$, $C_\gamma := 128 C_\chi^2$.
\end{theorem}

This theorem is a variant of Theorem~3.8 of
\cite{bacchefrosig:smo}, with additional information from the
proof of that theorem, in particular from Equations~(3.104),
(3.107) and (3.109). Another difference is due to our different
definition of the Renormalization transformation, i.e., without
analytic deformation of the spectral parameter.


\section{Renormalization Preserves Analyticity}
\label{sec:analyt-fesh}

This section provides one of the key tools for our method to work,
Proposition~\ref{analytic00} below, which implies that analyticity is
preserved under renormalization. It is part (a) of the following
proposition that is nontrivial and not proved in the papers of Bach
et al. (see Theorem~2.5 of \cite{bacchefrosig:smo} and the remark thereafter).

\begin{prop} \label{analytic00}
Let $S$ be an open subset of $\C^{\nu+1}$, $\nu\geq 0$. Suppose $\sigma\mapsto
H(w^{\sigma}) \in \mathcal{L}(\HH_{\rm red})$ is 
analytic on $S$, and that $H(w^{\sigma})$ belongs to some ball
$\BB(\alpha,\beta,\gamma)$ for all $\sigma\in S$. Then:
\begin{itemize}
\item[(a)] $H_{0,0}(w^{\sigma})$ is analytic on $S$.
\item[(b)] If for all $\sigma\in S$, $(H(w^{\sigma}),H_{0,0}(w^{\sigma}))$ is a Feshbach pair
for $\chi_\rho$, then $F_{\chi_\rho}(H(w^{\sigma}),H_{0,0}(w^{\sigma}))$ is
analytic on $S$.
\end{itemize}
\end{prop}

\begin{proof}
Suppose (a) holds true. Then $H_{0,0}(w^{\sigma})$ and
$W = H(w^{\sigma}) - H_{0,0}(w^{\sigma})$ are analytic function of $\sigma\in
S$ and hence so is the Feshbach map
$$
F_{\chi_\rho}(H(w^{\sigma}),H_{0,0}(w^{\sigma})) =  H_{0,0}(w^{\sigma}) + \chi_\rho W \chi_\rho
- \chi_\rho W \chib_\rho \left( H_{0,0}(w^{\sigma}) +
\chib_\rho W \chib_\rho \right)^{-1}  \chib_\rho W \chi_\rho.
$$
This proves (b) and it remains to prove (a).

Recall from Section~\ref{sec:renorm} that 
$B=\{k\in\R^3\times\{1,2\}:|k|\leq 1\}$ and 
let $P_1$ denote the projection onto the one boson subspace of
$\mathcal{H}_{\rm red}$, which is isomorphic to $L^2(B)$. Then
$P_1 H(w^{\sigma}) P_1$, like $H(w^{\sigma})$, is analytic and
\begin{eqnarray}
P_1 H(w^{\sigma}) P_1 &=& P_1 H_{0,0}(w^{\sigma}) P_1 + P_1 H_{1,1}(w^{\sigma}) P_1 \nonumber \\
&=& D_{\sigma} + K_{\sigma}  \label{eq:h00star} \; ,
\end{eqnarray}
where $D_{\sigma}$ denotes multiplication with $w^{\sigma}_{0,0}$
and $K_{\sigma}$ is the Hilbert Schmidt operator with kernel
$$
M_{\sigma}(k, \tilde{k} ) =  w^{\sigma}_{1,1}(0,k, \tilde{k} )\; .
$$
Our strategy is to show first that $K_{\sigma}$ and hence $P_1
H_{0,0}(w^{\sigma}) P_1 = P_1 H(w^{\sigma}) P_1 - K_{\sigma}$ is
analytic. Then we show that
$H_{0,0}(w^{\sigma})$ is an analytic operator on $\HH_{\rm red}$.\\

\noindent\underline{Step 1}: $K_{\sigma}$ is analytic.

For each $n \in \N$ let $\{ Q_i^{(n)}\}_i$ be a collection of $n$
measurable subsets of $B$ such that
\begin{equation} \label{eq:h001}
     B = \bigcup_{i=1}^n Q_i^{(n)} \ , \quad Q^{(n)}_i \cap Q^{(n)}_j =
     \emptyset, \ \ i \neq j \; ,
\end{equation}
and
\begin{equation} \label{eq:h002}
   |Q^{(n)}_i | \leq \frac{\rm const}{n} \; .
\end{equation}
Let $\chi_i^{(n)}$ denote the operator on $L^2(B)$ of multiplication
with $\chi_{Q^{(n)}_i}$. Then for $i\neq j$, $\chi_i^{(n)}
D_{\sigma} \chi_j^{(n)} = 0$ because $\chi_i^{(n)}$ and
$\chi_j^{(n)}$ have disjoint support and commute with $D_{\sigma}$.
Together with (\ref{eq:h00star}) this implies that
$$
\chi^{(n)}_i K_{\sigma} \chi^{(n)}_j = \chi^{(n)}_i P_1
H(w^{\sigma}) P_1 \chi_j^{(n)} \ , \quad {\rm for } \ \ i \neq j  \;
.
$$
Since the right hand side is analytic, so is the left hand side
and hence
$$
K^{(n)}_{\sigma}  = \sum_{i \neq j} \chi^{(n)}_i K_{\sigma}
\chi^{(n)}_j \;
$$
is analytic. It follows that $\sigma \mapsto \sprod{\varphi}{
K^{(n)}_{\sigma}\psi}$ is analytic for all $\varphi, \psi$ in
$L^2(B)$. Now let $\varphi, \psi \in C(B)$. Then
\begin{eqnarray*}
\lefteqn{ \left| \sprod{\varphi}{K^{(n)}_{\sigma} \psi } -
\sprod{\varphi}{ K_{\sigma} \psi }
\right| } \\
&= & \left|\int_{B\times B} \overline{\varphi}(x) \psi(y)
M_{\sigma}(x,y)\sum_{i=1}^n\chi^{(n)}_i(x) \chi^{(n)}_i(y) dx dy\right| \\
&\leq & \| \varphi \|_\infty \|\psi \|_\infty \| K_{\sigma} \|_{\rm
HS} \left(\sum_{i=1}^n |Q_i^{(n)}|^2\right)^{1/2} \longrightarrow 0
\; , \quad ( n \to \infty ),
\end{eqnarray*}
uniformly in $\sigma$, because the Hilbert Schmidt norm
$\|K_{\sigma} \|_{\rm HS}$ is bounded uniformly in $\sigma$ (in
fact, it is bounded by $\gamma$). This proves that $\sprod{\varphi}{
K_{\sigma} \psi}$ is analytic for all $\varphi, \psi \in C(B)$.
Since $C(B)$ is dense in $L^2(B)$, an other approximate argument
using $\sup_{\sigma} \| K_{\sigma} \| < \infty$ shows that
$\sprod{\varphi}{ K_{\sigma} \psi}$ is analytic for
all $\varphi, \psi \in L^2(B)$. Therefore $\sigma \mapsto K_{\sigma}$ is analytic \cite{kat:per}.\\


\noindent \underline{Step 2}: For each $k \in B$,
$w^{\sigma}_{0,0}(|k|)$ is an analytic function of $\sigma$.

For each $n\in\N$ let $f_{k,n}\in L^2(B)$ denote a multiple of the
characteristic function of $B_{1/n}(k) \cap B$ with $\| f_{n,k} \|
= 1$. By the continuity of $w^{\sigma}_{0,0}(|k|)$ as a function of $k$
\begin{eqnarray}
w^{\sigma}_{0,0}(|k|) &=& \lim_{n \to \infty} \int_B | f_{k,n}(x) |^2
w^{\sigma}_{0,0}(|x|) dx\label{eq:h001a} \\
&=& \lim_{n \to \infty} \langle a^*(f_{k,n}) \Omega, H_{0,0}(w^{\sigma})
a^*(f_{k,n}) \Omega \rangle.\nonumber
\end{eqnarray}
Since $a^*(f_{k,n}) \Omega \in P_1 \mathcal{H}_{\rm red}$ the
expression $\langle\cdots\rangle$, before taking the limit, is an
analytic function of $\sigma$. By assumption on $w^{\sigma}_{0,0}$,
this function is Lipschitz continuous with respect to $|k|$
\emph{uniformly in $\sigma$}. Therefore the convergence in
(\ref{eq:h001a}) is uniform in $\sigma$ and hence
$w^{\sigma}_{0,0}(|k|)$ is analytic by the Weierstrass
approximation theorem from complex analysis.\\

\noindent\underline{Step 3}: $H_{0,0}(w^{\sigma}) =
w^{\sigma}_{0,0}(H_f)$ is analytic.

By the spectral theorem
$$
\langle \varphi, w^{\sigma}_{0,0}(H_f P_{\rm red}) \varphi \rangle =
\int_{[0,1]} w^{\sigma}_{0,0}(\lambda) d \mu_{\varphi} (\lambda) \; .
$$
By an application of Lebesgue's dominated convergence theorem, using
$\sup_{\sigma} \|w^{\sigma}_{0,0} \| < \infty$, we see that the
right hand side, we call it $\varphi(\sigma)$, it is a continuous
function of $\sigma$. Therefore
$$
\int_\Gamma \varphi(\sigma) d\sigma =
\int_{[0,1]}\left(\int_{\Gamma}w^{\sigma}_{0,0}(\lambda) d\sigma\right) d \mu_{\varphi} (\lambda)
$$
for all closed loops $\Gamma:t\mapsto \sigma(t)$ in $S$, with $\sigma_j$ constant for
all but one $j\in\{1,\ldots,\nu+1\}$. The analyticity of $\sigma\mapsto\ph(\sigma)$ now follows from the analyticity of
$w^{\sigma}_{0,0}(\lambda)$ and the theorems of Cauchy and Morera.
By polarization, $w^{\sigma}_{0,0}(H_f P_{\rm red})$ is weakly analytic and hence
analytic.

\end{proof}


\section{Iterating the Renormalization Transform}
\label{sec:iterate}

In Section~\ref{sec:effect} we have reduced, for small $|g|$, the
problem of finding an eigenvalue of $H_g(s)$ in the neighborhood
$U_0(s):=\{z\in\C|(s,z)\in\UU\}$ of $\Eat(s)$ to finding $z\in \C$
such that $H^{(0)}[s,z]$ has a non-trivial kernel. We now use the
renormalization map to define a sequence $H^{(n)}[s,z]:=\RR^n H^{(0)}[s,z]$ of
operators on $\HH_{\rm red}$, which, by Theorem~\ref{sfm}, are isospectral in
the sense that $\ker H^{(n+1)}[s,z]$ is isomorphic to $\ker
H^{(n)}[s,z]$. The main purpose of the present section is to show that the operators $H^{(n)}[s,z]$ are well-defined
for all $z$ from non-empty, but shrinking sets
$U_n(s)\searrow\{z_{\infty}(s)\}$, $(n\to\infty)$. In the next section
it will turn out that $H^{(n)}[s,z_{\infty}(s)]$ has a non-trivial
kernel and hence that $z_{\infty}(s)$ is an eigenvalue of $H_g(s)$. The
construction of the sets $U_n(s)$ is based on Theorems~\ref{thm:fesh} and \ref{thm:sigal}, but not on the
explicit form of $H^{(0)}[s,z]$ as given by \eqref{eq:defh0}. Moreover, this construction is pointwise in $s$ and $g$,
all estimates being \emph{uniform in} $s\in V$ and $|g|<g_0$ for
some $g_0>0$. We therefore drop these parameters from our notations and we
now explain the construction of $H^{(n)}[z]$ making only the following assumption:
\begin{itemize}
\item[\textbf{(A)}] $U_0$ is an open subset of $\C$ and for every $z\in U_0$,
$$
    H^{(0)}[z] \in \BB(\infty,\rho/8,\rho/8).
$$
The polydisc $\BB(\infty,\rho/8,\rho/8)\subset
H(\WW_{\xi})$ is defined in terms of $\xi := \sqrt{\rho}/(4 C_\chi)$ and
$\mu>0$, where $\rho \in (0,1)$ and $C_\chi$ is given by Theorem~\ref{bcfssigal}.
\end{itemize}
By Lemma~\ref{feshbachtest}, we may define $H^{(1)}[z],\dots,H^{(N)}[z]$, recursively by
\begin{equation}\label{eq:def-Hn}
    H^{(n)}[z]:= \RR_{\rho}(H^{(n-1)}[z])
\end{equation}
provided that $H^{(0)}[z],\ldots, H^{(N-1)}[z]$ belong to
$\BB(\rho/2,\rho/8,\rho/8)$. Theorem~\ref{bcfssigal} gives us
sufficient conditions for this to occur: by iterating the map
$(\beta,\gamma)\mapsto (\beta',\gamma')$ starting with
$(\beta_0,\gamma_0)$, we find the conditions
\begin{eqnarray} \label{eq:cond1}
  \gamma_n := \left(C_\gamma \rho^{\mu}\right)^n \gamma_0 &\leq&  \rho/(8 C_\chi) \\
  \beta_n := \beta_0 + \left( \frac{C_\beta}{\rho} \sum_{k=0}^{n-1}
  (C_\gamma \rho^\mu)^{2k} \right) \gamma_0^2 &\leq& \rho/(8 C_\chi) \; ,
  \label{eq:cond2}
\end{eqnarray}
for $n=0,\ldots, N-1$. They are obviously satisfied for all $n\in\N$
if $C_\gamma \rho^\mu<1$ and if $\beta_0,\gamma_0$ are sufficiently
small. Let this be the case and let
$$E^{(n)}(z):=\sprod{\Omega}{H^{(n)}[z]\Omega}.$$ Then it remains to
make sure that $|E^{(n)}(z)|\leq \rho/2$ for $n=0,\ldots, N-1$. This
is achieved by adjusting the admissible values of $z$ step by step.
We define recursively, for all $n\geq 1$,
\begin{equation}\label{def:Un}
   U_n := \{ z\in U_{n-1}: |E^{(n-1)}(z) | \leq \rho/2\}.
\end{equation}
If $z\in U_N$, $H^{(0)}(z)\in \BB(\infty,\beta_{0},\gamma_{0})$,
and $\rho,\beta_0,\gamma_0$ are small enough, as explained above, then the
operators $H^{(n)}(z)$ for $n=1,\ldots,N$ are well defined by
\eqref{eq:def-Hn}. In addition we know from Theorem~\ref{bcfssigal}
that $H^{(n)}(z)\in \BB(\infty,\beta_n,\gamma_n)$, and that
\begin{equation} \label{eq:alphan3}
   \left|E^{(n)}(z)-\frac{E^{(n-1)}(z)}{\rho}\right| \leq
   \frac{C_\beta}{\rho}\gamma_{n-1}^2 =: \alpha_n.
\end{equation}
This latter information will be used in the proof of
Lemma~\ref{prop:balls} to show that the sets $U_n$ are not empty.
We summarize:

\begin{lemma} \label{cor:bcfs}
Suppose that (A) holds with $\rho\in(0,1)$ so small, that $C_\gamma \rho^\mu <1$.
Suppose $\beta_0, \gamma_0 \leq \rho/(8C_\chi)$ and, in addition,
\begin{equation} \label{eq:4446}
\beta_0 + \frac{C_\beta/\rho}{1 - ( C_\gamma \rho^\mu)^2}
\gamma_0^2 \leq \frac{\rho}{8 C_\chi} \; .
\end{equation}
If $H^{(0)}[z]\in \BB(\infty,\beta_0,\gamma_0)$ for all $z\in U_0$,
then $H^{(n)}[z]$ is well defined for $z\in U_n$, and
\begin{equation} \label{eq:basicstern}
   H^{(n)}[z] - \frac{1}{\rho} E^{(n-1)}(z) \in \BB(\alpha_n , \beta_n ,
\gamma_n), \quad {\it for} \ \ n \geq 1
\end{equation}
with $\alpha_n$, $\beta_n$, and $\gamma_n$ as in \eqref{eq:cond1},
\eqref{eq:cond2}, and  \eqref{eq:alphan3}.
\end{lemma}


The next lemma establishes conditions under which the
set $U_0$ and $U_n$ are non-empty. We introduce  the
discs
$$
 D_r := \{z \in \mathbb{C} | |z| \leq r \}
$$
and note that $U_n={E^{(n-1)}}^{-1}(D_{\rho/2})$.

\begin{lemma} \label{prop:balls}
Suppose that (A) holds with $U_0\ni \Eat$ and
$\rho \in (0,4/5)$ so small that $C_\gamma\rho^\mu <1$ and
$\overline{B(\Eat,\rho)}\subset U_0$. Suppose that $\alpha_0 <
\rho/2$, $\beta_0, \gamma_0 \leq \rho/(8C_\chi)$ and that
\eqref{eq:4446} hold. If $z\mapsto H^{(0)}[z]\in \LL(\HH_{\rm at})$ is
analytic in $U_0$ and $H^{(0)}[z]-(\Eat-z)\in\BB(\alpha_0,\beta_0,\gamma_0)$ for all $z\in U_0$, then the following is true.
\begin{itemize}
\item[(a)]
For $n \geq 0$, $E^{(n)}:U_n \to \C$ is analytic in $U_n^{\circ}$ and a
conformal map from $U_{n+1}$ onto $D_{\rho/2}$. In particular, $E^{(n)}$ has a unique
zero, $z_n$, in $U_n$. Moreover,
$$
B(E_{\rm at}, \rho) \supset U_1 \supset U_2 \supset U_{3} \supset \cdots \; .
$$
\item[(b)] The limit $z_\infty := \lim_{n\to\infty} z_n$ exists and  for
$\epsilon := 1/2 - \rho/2 - \alpha_1 > 0$,
$$
| z_n - z_\infty | \leq {\rho^{n}}  \exp\left( \frac{1}{2
 \rho \epsilon^2} \sum_{k=0}^\infty \alpha_k \right) \; .
$$
\end{itemize}
\end{lemma}

\noindent {\it Remark.} We call a function $f:A\to B$ \emph{conformal} if
it is the restriction of an analytic bijection $f:U\to V$ between open sets $U\supset A$ and $V\supset B$, and $f(A)=B$.

\begin{proof}
Since $H^{(0)}$ is analytic on $U_0$, it follows, by  Theorem
\ref{analytic00}, that $H^{(n)}$ is analytic on $U_n^\circ$ for all $n\in\N$.
In particular $E^{(n)}$ is analytic on $U_n^\circ$. To begin with we prove:
\begin{quote}
  $\boldsymbol{(I'_1)}$ $\quad U_{1} \subset B(E_{\rm at},\rho)$ and
  $E^{(0)}:U_1 \to D_{\rho/2}$ conformally.
\end{quote}
By assumption on $H^{(0)}(z)$,
\begin{equation} \label{eq:izero1}
|E^{(0)}(z) - (E_{\rm at} - z) | \leq \alpha_0 , \qquad \forall z
\in U_0 .
\end{equation}
Hence, if $z \in {E^{(0)}}^{-1}(D^\circ_{\rho/2 + \epsilon})$ then
$$
|E_{\rm at} - z | \leq \alpha_0 + \rho/2 + \epsilon < \rho \; ,
$$
provided $\epsilon >0$ is chosen sufficiently small. This proves
that $U_1 \subset {E^{(0)}}^{-1}(D^\circ_{\rho/2 + \epsilon})
\subset B(E_{\rm at},\rho)$. Since $E^{(0)}$ is continuous, it
follows that ${E^{(0)}}^{-1}(D^\circ_{\rho/2 + \epsilon})$ is open
in $\C$. If
\begin{equation} \label{eq:izero2}
E^{(0)} : {E^{(0)}}^{-1}(D^\circ_{\rho/2 + \epsilon} )
 \to D^\circ_{\rho/2 + \epsilon}  \qquad \text{ is a bijection}
,
\end{equation}
then it is conformal on $U_1$. So it suffices to prove
\eqref{eq:izero2}. To this end we use Rouche's theorem. Let $w \in
D^\circ_{\rho/2 + \epsilon}$. Then $E_{\rm at} - z - w$ has exactly
one zero $z \in B(E_{\rm at},\rho)$ and for all $z \in
\partial B(E_{\rm at},\rho)$,
$$
| E_{\rm at} - z - w | \geq \rho - |w| \geq \rho/2 > \alpha_0 .
$$
Since, by \eqref{eq:izero1},
$$
|( E^{(0)}(z) - w )  - (E_{\rm at} - z - w ) | \leq \alpha_0 ,
$$
for all $z \in \overline{B(E_{\rm at},\rho)}$, it follows that
$E^{(0)}(z) - w$, like $(E_{\rm at} - z - w )$ has exactly one zero
$z \in  \overline{B(E_{\rm at},\rho)}$. This proves
\eqref{eq:izero2} because $ {E^{(0)}}^{-1}(D^\circ_{\rho/2 +
\epsilon} ) \subset B(E_{\rm at},\rho)$.

Next we prove, by induction in $n$, that
\begin{quote}
  $\boldsymbol{(I_n)}$\quad $E^{(n-1)}:U_n \to D_{\rho/2}$ conformally.
\end{quote}
For $n=1$, this follows from $I_{1}'$. Suppose $I_{n}$, holds. First note that $\alpha_n \leq \alpha_1 =
(C_\beta /\rho) \gamma_0^2$, Ineq. \eqref{eq:4446}, $C_\chi \geq 1$,
and $\rho<4/5$ imply
\begin{equation} \label{eq:4445b}
\alpha_n  + \rho/2 < 1/2  \; .
\end{equation}
Thus we can choose a  positive $\epsilon$ such that
\begin{eqnarray} \label{eq:nullnull}
\alpha_n + \rho/2 + 2 \epsilon < 1/2 \; .
\end{eqnarray}
We define $D_+^\circ := D^\circ_{\rho/2 + \epsilon}$ and $D_-^\circ
:= D^\circ_{\rho/2 - \rho \epsilon}$, so that $D^\circ_- \subset
D_{\rho/2} \subset D_{+}^\circ$. We claim that
\begin{eqnarray} \label{eq:claim01}
 { E^{(n)}}^{-1} (D^\circ_+) \subset{E^{(n-1)}}^{-1}
 (D^\circ_-)
\end{eqnarray}
and that
\begin{eqnarray} \label{eq:claim02}
E^{(n)} :  { E^{(n)}}^{-1} (D^\circ_+) \to  (D^\circ_+)  \quad
\text{ is a bijection.}
\end{eqnarray}
Suppose  \eqref{eq:claim01} and \eqref{eq:claim02} hold. Then by
\eqref{eq:claim01} and the induction Hypothesis $I_n$,
${E^{(n)}}^{-1}(D_+^\circ) \subset U_{n}^\circ$. Since $E^{(n)}$ is
continuous on $U_n^\circ$, it follows that
${E^{(n)}}^{-1}(D_+^\circ)$ is open. Since $E^{(n)}$ is analytic,
\eqref{eq:claim02} implies $I_{n+1}$. It remains to prove
\eqref{eq:claim01} and \eqref{eq:claim02}.

\eqref{eq:claim01} follows from  \eqref{eq:basicstern} and
\eqref{eq:nullnull}: if $|E^{(n)}(z)| < \rho/2 + \epsilon$ and
$|E^{(n)}(z) - \rho^{-1} E^{(n-1)}(z) | \leq \alpha_n$, then
$|E^{(n-1)}(z)| < \rho/2 - \rho\epsilon$.

To prove \eqref{eq:claim02} we use Rouche's Theorem. Let $w \in
D^\circ_+$. Then, by \eqref{eq:nullnull}, $\rho w \in D^\circ_-$ and
the induction Hypothesis $I_n$ implies that $E^{(n-1)}(z) - \rho w$
has exactly one zero $z \in {E^{(n-1)}}^{-1}(D_-^\circ)$. On the
other hand, by \eqref{eq:nullnull},
$$
|\rho^{-1}(E^{(n-1)}(z) - \rho w ) | \geq \rho^{-1} | E^{(n-1)}(z)|
> \alpha_{n}  , \qquad \forall z \in
\partial ({E^{(n-1)}}^{-1}(D_-^\circ))  \; .
$$
Since, by  \eqref{eq:basicstern},
$$
|(E^{(n)}(z) - w ) - \rho^{-1}(E^{(n-1)}(z) - \rho w ) | \leq
\alpha_{n} \quad , \quad \forall z \in U_{n} ,
$$
it follows that $E^{(n)}(z) -w$, like $E^{(n-1)}(z) - \rho w$, has
exactly one zero $z \in {E^{(n-1)}}^{-1}(D_-^\circ)$. Therefore,
\eqref{eq:claim02}  follows from \eqref{eq:claim01}.

(b) By (a), $U_{k+1}$ contains $z_k$ and all subsequent terms of the
sequence $(z_n)_{n =1}^\infty$. Thus, to prove that
$(z_n)_{n=1}^\infty$ converges, it suffices to show that the
diameter of $U_n$ tends to zero as $n$ tends to infinity. To this
end, let $F^{(k)}$ denote the inverse of the function $E^{(k)}:
U_{k+1} \to D_{\rho/2}$. Then
\begin{eqnarray} \label{eq:diameter}
 {\rm diam}( U_{n+1})  &=&  {\rm diam} (F^{(n)}(D_{\rho/2}  ) )  \nonumber \\
 &=& {\rm diam} \left( E^{({\rm at})} \circ F^{(0)} \circ E^{(0)}  \cdots \circ F^{(n-1)}
\circ E^{(n-1)} \circ F^{(n)} (D_{\rho/2} ) \right) \; ,
\end{eqnarray}
where we used that $z \mapsto E^{({\rm at})}(z) := E_{\rm at} - z$
is an isometry. We want to estimate \eqref{eq:diameter} from above.
Let $k \geq 1$. For all $z \in D_{\rho/2}$, by (\ref{eq:basicstern}),
\begin{equation} \label{eq:sternstern}
\left|  \rho z - E^{(k-1)}\left(F^{(k)}(z)\right)\right|\leq\rho\alpha_k,
\end{equation}
and hence $|E^{(k-1)}\circ F^{(k)}(z)| \leq\rho \alpha_k + \rho^2/2
\leq \rho/2 -  \epsilon \rho$, where $\epsilon := 1 /2 -  \rho/2 -
\alpha_1$  is positive by \eqref{eq:4445b}. This shows that
$E^{(k-1)}\circ F^{(k)}$ maps $D_{\rho/2}$ into $D_{{\rho/2} - \rho
\epsilon}$. By Cauchy's integral formula and by
\eqref{eq:sternstern},
\begin{eqnarray}  \label{eq:jetzt}
\big|\partial_z ( E^{(k-1)} \circ F^{(k)}(z)  - \rho z ) \big| &=& \left|
\frac{1}{2 \pi i } \int_{\partial D_{\rho/2}} \frac{E^{(k-1)}\circ
F^{(k)}(w) - \rho w}{(z - w)^2} dw  \right| \\
\nonumber &\leq& \alpha_k  /(2 \epsilon^2),\qquad\qquad\text{for}\
z\in D_{{\rho/2}- \rho \epsilon}.
\end{eqnarray}
It follows that $| ( E^{(k-1)} \circ F^{(k)} )'(z)|= \rho + \alpha_k
/(2 \epsilon^2)$ for $z\in D_{{\rho/2}-  \rho \epsilon}$. A similar
estimate yields $|(E^{({\rm at})}\circ F^{(0)})'(z)| \leq 1 +
\alpha_0 /(2 \rho \epsilon^2)$ for $z \in D_{\rho/2 - \rho
\epsilon}$.
 Using these estimates and \eqref{eq:diameter} we obtain
\begin{eqnarray*}
 {\rm diam}( U_{n+1} )  &\leq&   ( 1 + \alpha_0 /(2 \rho
 \epsilon^2))
{\rm diam}( E^{(0)} \circ F^{(1)} \circ \cdots \circ F^{(n-1)}(D_{{\rho/2} -  \rho \epsilon}) ) \\
& \leq& \rho^{n-1} \prod_{k=0}^{n-1} \left( 1 + \alpha_k /(2 \rho
\epsilon^2)\right)
{\rm diam} D_{\rho/2-  \rho \epsilon }  \\
& \leq& \rho^{n} \exp \left( \sum_{k=0}^\infty \alpha_k /(2 \rho
\epsilon^2) \right)   \; ,
\end{eqnarray*}
where we used that $1 + x \leq \exp(x)$ in the last inequality. This proves (b).
\end{proof}

The following results will allow us to show that $z_\infty(s)=\inf\sigma(H(s))$,  if  $s\in \R$.

\begin{cor} \label{cor:derivative} Suppose the assumptions of Lemma \ref{prop:balls}
hold, $E_{\rm at}\in\R$, and $H^{(0)}(z)^* =
H^{(0)}(\overline{z})$ for all $z \in \overline{B(E_{\rm
at},\rho)}$. Then for all $n\geq 0$, $ U_{n+1} \cap \R$ is an
interval and $\partial_x E^{(n)}(x) < 0$  on $U_{n+1} \cap \R$.
\end{cor}

\begin{proof}  Using an induction argument and the definition of the
renormalization transformation one sees that $H^{(n)}(z)^* =
H^{(n)}(\overline{z})$ for $z \in U_n$. In particular,
\begin{equation} \nonumber
\overline{E^{(n)}(z)} = E^{(n)}(\overline{z})\qquad \text{for all}\ z \in
U_n .
\end{equation}
This together with $E^{(n)}:U_{n+1} \to D_{\rho/2}$ being a homeomorphism, c.f. Lemma~\ref{prop:balls}, implies that
$$
[a_{n+1}, b_{n+1} ]:=( E^{(n)})^{-1}[-\rho/2,\rho/2] =  U_{n+1} \cap \R
$$
is indeed an interval. Moreover, by Lemma \ref{prop:balls},
$$
 E_{\rm at} - \rho < a_1 < a_2 < ... \leq z_\infty \; .
$$
We prove by induction  that for all $n \in \N$,
\begin{equation}\label{eq:indhyp}
\partial_x E^{(n)}(x) < 0 \qquad\mathrm{ on} \ \ [a_{n+1}, b_{n+1}] \; .
\end{equation}

We begin with $n=0$. By assumption on $H^{(0)}[z]$, $| E^{(0)}(z) -
(E_{\rm at} - z ) | \leq \alpha_0 $ for $z \in U_0$. For $z = E_{\rm
at} - \rho$, which belongs to $U_0$ by choice of $\rho$, we obtain
$$
|E^{(0)}(E_{\rm at} - \rho ) - \rho | \leq \alpha_0 < \frac{1}{2}
\rho \; ,
$$
by assumption on $\alpha_0$. This proves that $E^{(0)}(E_{\rm at} -
\rho)
> \rho/2$. Since $|E^{(0)}(x)| \geq \rho/2$ for $x \in [E_{\rm at} -
\rho, a_1]$ the function $E^{(0)}$ must be positive on this
interval. On the other hand it is a diffeomorphism from $[a_1,b_1]$
onto $[-\rho/2,\rho/2]$ by Lemma \ref{prop:balls}. It follows that
$\partial_x E^{(0)}(x) < 0$ for $x \in [a_1,b_1]$.

To prove \eqref{eq:indhyp} for $n \geq 1$ suppose that
\begin{equation}\label{eq:indhyp1}
\partial_x E^{(n-1)}(x) < 0 \qquad \mathrm{ on} \ [a_{n},b_{n}] \;
.
\end{equation}
Let $F^{(n)}$ be the inverse of $E^{(n)}: U_{n+1} \to D_{\rho/2}$.
Setting $z=0$ in \eqref{eq:jetzt} we obtain
\begin{eqnarray*}
\left| \left. \partial_x \left( E^{(n-1)} \circ F^{(n)}(x) - \rho x
\right) \right|_{x=0} \right| \leq \frac{\rho}{2} \frac{\rho
\alpha_n}{(\rho/2)^2}\leq 2 \alpha_1 < \rho \; .
\end{eqnarray*}
This shows that
\begin{eqnarray*}
  0 &<&   \left( E^{(n-1)} \circ F^{(n)} \right)'(0) \\
   &=&
   \left( \partial_x  E^{(n-1)} \right) ( F^{(n)}(0))  \frac{1}{ ( \partial_x  E^{(n)})(
   F^{(n)}(0))} \; .
\end{eqnarray*}
Hence $( \partial_x  E^{(n)}) ( F^{(n)}(0)) $ has the same sign as
$( \partial_x  E^{(n-1)})( F^{(n)}(0))$, which is negative by
induction hypothesis \eqref{eq:indhyp1}. Since $E^{(n)}:
[a_{n+1},b_{n+1}] \to [-\rho/2, \rho/2]$ is a diffeomorphism,
$\partial_x E^{(n)}(x) < 0$ for all $x \in [a_{n+1},b_{n+1}]$.
\end{proof}

\begin{prop} \label{spectralgap}
Suppose the  assumptions of  Lemma~\ref{prop:balls} are satisfied,
$E_{\rm at}$ is real and  $H^{(0)}[z]^* = H^{(0)}[\overline{z}]$ for
$z \in \overline{B(E_{\rm at},\rho)}$. Then, there exists an
$a<z_{\infty}$ such that $H^{(0)}[x]$ has a bounded inverse for
$x\in (a,z_{\infty})$.
\end{prop}

\begin{proof}
Let $[a_n , b_n ] = U_n \cap \R$, c.f. Corollary
\ref{cor:derivative}. Then, by Lemma \ref{prop:balls}, $a_1 < a_2  <
a_3 < ... < z_\infty$ and $\lim_{n \to \infty}a_n  = z_\infty$. We
show that $H^{(n)}[x]$ is bounded invertible for $x\in
[a_n,a_{n+1})$. By a repeated application of the Feshbach property,
Theorem \ref{sfm} {\it (i)}, it will follow that $H^{(n-1)}[x], ...
, H^{(0)}[x]$ are also bounded invertible for $x\in [a_n,a_{n+1})$.

Let $x\in [a_n,a_{n+1})$. Then both $H^{(n)}[x]$ and $H^{(n)}_{0,0}[x]$
are self-adjoint and, by \eqref{eq:westimate} and
\eqref{eq:basicstern},
\begin{eqnarray} \label{eq:lowerbound}
H^{(n)}[x] = H^{(n)}_{0,0}[x] + ( H^{(n)}[x] - H^{(n)}_{0,0}[x] )
 \geq E^{(n)}(x) - \xi \gamma_n,
\end{eqnarray}
where we have used that $H^{(n)}_{0,0}[x] \geq E^{(n)}(x)$, which
follows from $\beta_n < 1$. Since the function $E^{(n)}$ is
decreasing on $[a_{n+1},b_{n+1}]$ with a zero in this interval, we
know that $E^{(n)}(a_{n+1})>0$. On the other hand, by construction
of $U_n$, $|E^{(n)}|\geq \rho/2$ on $[a_{n},a_{n+1})$. Therefore
\eqref{eq:lowerbound} implies that $H^{(n)}[x]\geq
(\rho/2-\xi\gamma_n)>(\rho/2-\xi\rho/8)>0$, which proves that
$H^{(n)}[x]$ is bounded invertible.
\end{proof}


\section{Construction of the Eigenvector}
\label{sec:ev}

Next we show that zero is an eigenvalue of $H^{(0)}[z_{\infty}]$. In fact, we
will show that zero is an eigenvalue of
$H^{(n)}[z_{\infty}]$ for every $n\in\N$. To this
end we define
$$
   Q_n[z] := \chi_\rho - \overline{\chi}_\rho
   {\left(H^{(n)}_{\overline{\chi}_\rho}[z]\right)}^{-1} \overline{\chi}_\rho
   W^{(n)}[z]
   \chi_\rho,\qquad\text{for}\ z\in U_n,
$$
where $W^{(n)} = H^{(n)} - H^{(n)}_{0,0}$.  By the definition of
$H^{(n)}[z]$ and by Lemma~\ref{lm:F-basics} (c),
\begin{equation}\label{eq:H(n-1)Hn}
  H^{(n-1)}[z] Q_{n-1}[z] \Gamma^*_\rho =
  \big(\rho\Gamma_{\rho}^{*}\chi_1\big)H^{(n)}[z]
\end{equation}
and moreover, if $H^{(n)}[z]\ph=0$ and $\ph\neq 0$ then
$Q_{n-1}[z]\Gamma^*_\rho\ph\neq 0$ by Theorem~\ref{sfm}. Thus if $0$ is an eigenvalue of
$H^{(n)}[z]$, then it is an eigenvalue of $H^{(n-1)}[z]$ as well,
and the operator $Q_{n-1}[z] \Gamma^*_\rho$ maps the corresponding
eigenvectors of $H^{(n)}[z]$ to eigenvectors of $H^{(n-1)}[z]$.


\begin{theorem}\label{thm:vector}
Suppose the assumptions of Lemma~\ref{prop:balls}  hold. Then the
limit
$$
\ph^{(0)} = \lim_{n \to \infty} Q_0[z_\infty] \Gamma_\rho^*
Q_1[z_\infty] ... \Gamma_\rho^* Q_n[z_\infty] \Omega
$$
exists, $\ph^{(0)} \neq 0$ and $H^{(0)}[z_{\infty}] \varphi^{(0)}=
0.$ Moreover,
$$
\left\| \varphi^{(0)} - Q_0[z_\infty] \Gamma_\rho^* Q_1[z_\infty]
... \Gamma_\rho^* Q_n[z_\infty] \Omega \right\| \leq C
\sum_{l=n+1}^\infty \gamma_l \; ,
$$
where $C = C(\rho,\xi,\gamma_0)$.
\end{theorem}

\noindent\emph{Remark.} By Theorem~\ref{thm:vector} and by Proposition \ref{pro:initial} {\it (ii)},
$Q_{\boldsymbol{\chi}}(\varphi_{\rm at} \otimes \varphi^{(0)})$ is
an eigenvector of $H_g$ with eigenvalue $z_\infty$.

\begin{proof}
For $k,l\in \N$ with $k\leq l$ we define $\ph_{k,l}\in \HH_{\rm red}$ by
$$
    \ph_{k,l}:= (Q_k[z_\infty]\Gamma_{\rho}^*)(Q_{k+1}[z_\infty]\Gamma_{\rho}^*)\cdot\ldots\cdot (Q_{l-1}[z_\infty]\Gamma_{\rho}^*)Q_l[z_\infty]\Omega.
$$
and we set $\ph_{k,k-1}:=\Omega$.

\noindent\underline{Step 1}: There is a constant $C<\infty$ depending on
$\xi,\rho$ and $\sum_{n}\gamma_n$ such that, for all $k,l\in \N$
with $k\leq l$
$$
    \|\ph_{k,l}-\ph_{k,l-1}\| \leq C \gamma_{l}
$$

By definition of $\ph_{k,l}$ and since
$\Omega=\Gamma_{\rho}^{*}\chi_{\rho}\Omega$
$$
   \ph_{k,l}-\ph_{k,l-1} = \prod_{n=k}^{l-1}
   (Q_n[z_\infty]\Gamma_{\rho}^*)(Q_{l}[z_\infty]-\chi_{\rho})\Omega.
$$
where the empty product in the case $k=l$ is to interpret as the
identity operator. Since on $U_n$,
$\|Q_n\Gamma_{\rho}^*\|=\|Q_n\|\leq \|Q_n-\chi_{\rho}\|+1\leq
\exp\|Q_n-\chi_{\rho}\|$ it follows that
\begin{equation}\label{eq:vec1}
  \|\ph_{k,l}-\ph_{k,l-1}\| \leq
   \exp\left(\sum_{n=k}^{l-1}\|Q_n[z_\infty]-\chi_{\rho}\|\right)\|Q_{l}[z_\infty]-\chi_{\rho}\|,
\end{equation}
and hence it remains to estimate $\|Q_n[z_\infty]-\chi_{\rho}\|$. By
definition of $Q_n$, on $U_n$,
\begin{equation}\label{eq:vec2}
   Q_n-\chi_{\rho} =
   -\overline{\chi}_\rho\big(H_{\overline{\chi}_\rho}^{(n)}\big)^{-1}
   \overline{\chi}_\rho\big(H^{(n)}-H^{(n)}_{0,0}\big)\chi_{\rho}
\end{equation}
and by estimates in the proof of Lemma~\ref{feshbachtest},
\begin{equation}\label{eq:vec3}
  \|\big(H_{\overline{\chi}_\rho}^{(n)}\big)^{-1}
   \overline{\chi}_\rho\| \leq
   \frac{8}{\rho}\frac{1}{1-\xi},\qquad
   \|H^{(n)}-H^{(n)}_{0,0}\| \leq \xi \gamma_{n}.
\end{equation}
Equation~\eqref{eq:vec2}, combined with the estimates
\eqref{eq:vec1}, and \eqref{eq:vec3} prove Step 1 with
\begin{equation*}
   C:= \frac{8}{\rho}\frac{\xi}{1-\xi}
   \exp\left(\frac{8}{\rho}\frac{\xi}{1-\xi}\sum_{n\geq 0}\gamma_n\right).\\
\end{equation*}

\noindent\underline{Step 2}: For all $k\in\N$, the limit
$$
    \ph_{k,\infty} := \lim_{n\to\infty} \ph_{k,n}
$$
exists, the convergence being uniform in $s$, and
$\ph_{k,\infty}\neq 0$ for $k$ sufficiently large.

Summing up the estimates from Step 1 for all $l$ with $l\geq n+1$ we
arrive at
$$
   \|\ph_{k,\infty}-\ph_{k,n}\| \leq C \sum _{l=n+1}^{\infty}\gamma_l \to0,\qquad n\to\infty,
$$
uniformly in $s$. Specializing this inequality to $n=k-1$ so that
$\ph_{k,n}=\ph_{k,k-1}=\Omega$, we see that
$\|\ph_{k,\infty}-\Omega\|<1=\|\Omega\|$ and hence
$\ph_{k,\infty}\neq 0$ for
sufficiently large $k$.\\

\noindent\underline{Step 3}: For all $k\in \N$,
$$
    H^{(k)}[z_{\infty}]\ph_{k,\infty}[z_{\infty}] = 0,\quad\text{and}\quad
    \ph_{k,\infty}[z_{\infty}]\neq 0.
$$

Since $H^{(k)}[z_\infty]$ is a bounded operator and by
\eqref{eq:H(n-1)Hn},
\begin{eqnarray}
H^{(k)}[z_\infty] \ph_{k,\infty}
&=& \lim_{n \to \infty} H^{(k)}[z_\infty] \varphi_{k,n}\nonumber \\
&=& \lim_{n\to\infty} \big(\rho \Gamma^*_\rho
\chi_1\big)^{n-k+1}H^{(n+1)}[z_\infty]\Omega.\label{eq:vec4}
\end{eqnarray}
Using $H^{(n+1)}[z_\infty]\Omega=E^{(n+1)}(z_\infty)\Omega +
(H^{(n+1)}[z_\infty]-H^{(n+1)}_{0,0}[z_\infty])\Omega$ and
$$
  |E^{(n+1)}(z_\infty)| \leq \frac{\rho}{2},\qquad
  \|H^{(n+1)}[z_\infty]-H^{(n+1)}_{0,0}[z_\infty]\|\leq \gamma_n\leq \gamma_0,
$$
we see that the limit \eqref{eq:vec4} vanishes because
$\lim_{n\to\infty}\rho^n=0$.

From $\ph_{k-1,n} = (Q_{k-1}[z_\infty]\Gamma_{\rho}^{*})\ph_{k,n}$,
the boundedness of the operator $Q_k$, and from Step~2 it follows that,
$$
  \ph_{k-1,\infty} = (Q_{k-1}[z_\infty]\Gamma_{\rho}^{*})\ph_{k,\infty}.
$$
Since $\ph_{k,\infty}$ belongs to the kernel of
$H^{(k)}[z_{\infty}]$, as we have just seen, it follows from
Theorem~\ref{sfm} that $\ph_{k-1,\infty}\neq 0$ whenever
$\ph_{k,\infty}\neq 0$. Iterating this argument starting with $k$ so
large that, by Step~2, $\ph_{k,\infty} \neq 0$, we conclude that
$\ph_{k,\infty}\neq 0$ for all $k\in \N$.
\end{proof}


\section{Analyticity of Eigenvalues and Eigenvectors}
\label{sec:analyt-ev}

This section is devoted to the proof of Theorem~\ref{thm:main}. It is
essential for this proof, that a neighborhoods $V_0\subset V$ of $s_0$ and a
bound $g_1$ on $g$ can be determined in such a way that the renormalization analysis of
Sections~\ref{sec:iterate} and \ref{sec:ev}, and in particular the choices of
$\rho$ and $\xi$ are independent of $s\in V_0$ and $g<g_1$. Once $V_0$ and $g$
are found, the assertions of Theorem~\ref{thm:main} are derived from
Proposition~\ref{analytic00} and the uniform bounds of Sections~\ref{sec:iterate} and \ref{sec:ev}.

\begin{proof}[Proof of Theorem \ref{thm:main}]
Let $\mu>0$ and $\UU\subset\C^{\nu + 1}$ be given by Hypothesis
\ref{hyp:G} and Hypothesis \ref{hyp:R}, respectively.  For the
renormalization procedure to work, we first choose $\rho\in(0,4/5)$
and a open neighborhood $V_0\subset V$ of $s_0$, both small enough,
so that $C_{\gamma}\rho^{\mu}<1$ and
\begin{equation} \label{eq:set}
\overline{B(E_{\rm at}(s), \rho)} \subset \{z|(s,z)\in
\UU\},\qquad \text{if}\ s\in V_0.
\end{equation}
This is possible since $s \mapsto E_{\rm at}(s)$ is continuous. Let
$\xi=\sqrt{\rho}/(4C_{\chi})$. Next we pick small positive constants
$\alpha_0$, $\beta_0$, and $\gamma_0$ such that
\begin{equation}\label{eq:abc}
   \alpha_0 < \frac{\rho}{2}, \qquad \beta_0 \leq \frac{\rho}{8C_{\chi}},
   \qquad \gamma_0 \leq \frac{\rho}{8C_{\chi}},
\end{equation}
and in addition
\begin{equation} \label{le:zinfty2}
   \beta_0 + \frac{C_{\beta}/\rho}{1-(C_{\chi}\rho^\mu)^2 }\gamma_0^2 \leq
   \frac{\rho}{8C_{\chi}}.
\end{equation}
By Theorems~\ref{thm:fesh} and \ref{thm:sigal}, there exists a $g_1>0$ such that for $|g| \leq g_1$
$$
H_g^{(0)}[s,z] - ( E_{\rm at}(s) - z ) \in \mathcal{B}(\alpha_0,
\beta_0, \gamma_0 ), \qquad \text{for}\ (s,z) \in \UU,
$$
where $H_g^{(0)}[s,z]$ is analytic on $\UU$. We define
\begin{eqnarray*}
   \UU_0 &:=& \UU \\
   \UU_n &:=& \{ (s,z) \in \UU_{n-1}: | E^{n-1}(s,z) | \leq \rho/8 \}.
\end{eqnarray*}
and
$$
    U_n(s) := \{ z| (s,z) \in \UU_n \}, \qquad n\in \N.
$$
Then, by \eqref{eq:abc}, \eqref{le:zinfty2}, and \eqref{eq:set} the
assumptions of Lemma~\ref{prop:balls} are satisfied for $s\in V_0$
and $U_0=U_{0}(s)$. It follows that, for all $n\in \N$,
$H^{(n)}[s,z]=\RR^n H^{(0)}[s,z]$ is well-defined for
$(s,z)\in \UU_n$, and that $U_n(s)\neq \emptyset$. By
Proposition~\ref{analytic00}, $H^{(n)}[s,z]$ is analytic in $ \UU_n^{\circ}$.\\

\noindent \underline{Step 1}: $z_\infty(s)=\lim_{n\to\infty}z_n(s)$
exists and is analytic on $V_0$.  \\

Since $H^{(n)}[s,z]$ is analytic on $\UU_n^{\circ}$, so is $E^{(n)}(s,z)$. Let $z_n(s)$ denote the
unique zero of the function $z\mapsto E_n(s,z)$ on $U_n(s)$ as
determined by Lemma~\ref{prop:balls}. That is,
$$
E^{(n)}(s,z_n(s)) = 0.
$$
By the implicit function theorem  $z_n(s)$ is analytic in $s$. The
application of the implicit function theorem is justified since $z
\mapsto E^{(n)}(s,z)$ is bijective in a neighborhood of $z_n(s)$,
and thus in this neighborhood $\partial_z E^{(n)}(s,z)\neq 0$. By
Lemma~\ref{prop:balls} {\it (b)}, $z_n(s)$ converges to
$z_\infty(s)$ uniformly in $s \in V_0$.  This implies the
analyticity of $z_\infty(s)$ on $V_0$, by the Weierstrass
approximation theorem of complex analysis.\\

\noindent \underline{Step 2}: For $s \in V_0$, there exists an
eigenvector $\psi(s)$ of $H(s)$ with eigenvalue $z_\infty(s)$, such
that $\psi(s)$  depends analytically on $s$. \\

 Since $H^{(n)}[s,z]$ is analytic on $\UU_n^{\circ}$, it follows, by Proposition~\ref{analytic00}, that
$$
   Q_n[s,z] = \chi_\rho(s) - \overline{\chi}_\rho(s)
   {H^{(n)}_{\overline{\chi}_\rho}[s,z]}^{-1} \overline{\chi}_\rho(s)
   W^{(n)}[s,z]
   \chi_\rho(s)
$$
is analytic  on $\UU_n^{\circ}$, where $W^{(n)} := H^{(n)} -
H^{(n)}_{0,0}$. Hence, by Step~1, $s \mapsto Q_n[s,z_\infty(s)]$ is analytic on $V_0$. It follows that
$$
\varphi_{0,n}(s):=Q_0[s,z_\infty(s)] \Gamma_\rho^* Q_1[s,z_\infty(s)]\dots
\Gamma_\rho^* Q_n[s,z_\infty(s)] \Omega
$$
is analytic on $V_0$. From Theorem~\ref{thm:vector} we know that
these vectors converge uniformly on $V_0$ to  a vector $
\ph^{(0)}(s)\neq 0$ and that $H^{(0)}[s,z_{\infty}(s)] \varphi^{(0)}(s)=
0$. Hence $\ph^{(0)}(s)$ is analytic on $V_0$ and, by the Feshbach property (Proposition \ref{pro:initial}{\it (ii)}), the vector
$$
\psi(s) = Q_{\boldsymbol{\chi}}[s,z_\infty(s)]\big( \varphi_{\rm
at}(s) \otimes \varphi^{(0)}(s)\big)
$$
is an eigenvector of $H(s)$ with eigenvalue with $z_\infty(s)$. Since
$\ph_{\rm at}$ is analytic on $V_0$ we conclude that $\psi$ is analytic on
$V_0$ as well.\\

\noindent\underline{Step 3}: For $s\in V_0\cap\R^{\nu}$, $z_\infty(s)={\rm
inf}\sigma(H(s))$. \\

Let $s \in V_0\cap\R^{\nu}$. Then $H(s)$ is self-adjoint and its spectrum is a
half line $[E(s),\infty)$. By Step~2, $z_{\infty}(s)\geq E(s)$. We use
Proposition~\ref{spectralgap} to show that $z_{\infty}(s)>E(s)$ is impossible.
Clearly $\Eat(s)\in\R$, and $H^{(0)}[s,z]^* = H^{(0)}[s,\overline{z}]$ for $z
\in B(E_{\rm at}(s),\rho)$ is a direct consequence of the definition of
$H^{(0)}$ and the self-adjointness of $H(s)$. Hence there exists a number $a(s) < z_\infty(s)$
such that $H^{(0)}[s,x]$ has a
bounded inverse for all $x \in (a(s) ,z_\infty(s))$. It follows, by
Theorem~\ref{sfm}, that  $(a(s),z_\infty(s))\cap\sigma(H(s))=\emptyset$. Therefore $z_{\infty}(s)=E(s)$.
\end{proof}

\appendix

\section{Neighborhood of Effective Hamiltonians}
\label{sec:neighbor}

The purpose of this section is to prove the following Theorem.

\begin{theorem} \label{thm:sigal}
Let Hypotheses \ref{hyp:G}, \ref{hyp:H}, and \ref{hyp:R} hold for
some $\mu>0$ and $\UU\subset\C\times\C$. For every $\xi\in (0,1)$
and every triple of positive constants $\alpha_0,\beta_0,
\gamma_0$, there exists a positive constant $g_1$ such that for
all $g\in [0,g_1)$ and all $(s,z) \in \UU$, $(H_g(s)-z,H_0(s)-z)$
is a Feshbach pair for $\boldsymbol{\chi}(s)$, and
\begin{equation} \label{eq:nb1}
   H_g^{(0)}[s,z] - (E_{\rm at}(s) - z ) \in
\mathcal{B}(\alpha_0,\beta_0, \gamma_0).
\end{equation}
\end{theorem}

\noindent By  Theorem \ref{thm:fesh} we know that we can choose $g$
sufficiently small such that the Feshbach property is satisfied.
To prove (\ref{eq:nb1}) we explicitly compute the sequence of kernels
$w=(w_{m,n})\in \WW_{\xi}$ such that $H_g^{(0)}[s,z]=H(w)$. To this end we
recall that, by \eqref{eq:defh0} and \eqref{eq:eff22},
\begin{eqnarray} \label{eq:nb2}
H_g^{(0)}[s,z] = (E_{\rm at} - z) + H_f +
 \langle \chi_1 (gW  - gW\overline{\boldsymbol{\chi}}
(H_g - z )^{-1}_{\overline{\boldsymbol{\chi}}}
\overline{\boldsymbol{\chi}}gW)\chi_1\rangle_{\rm at},
\end{eqnarray}
and we expand the resolvent $(H_g - z )^{-1}_{\overline{\boldsymbol{\chi}}}$
in a Neumann series. We find that
$$
\langle \chi_1  (g W  - g W \overline{\boldsymbol{\chi}}( H_g -
z)^{-1}_{\overline{\boldsymbol{\chi}}} \overline{\boldsymbol{\chi}}
g W ) \chi_1 \rangle_{\rm at} = \sum_{L=1}^\infty ( -1)^{L-1} g^L \langle\chi_1(WF)^{L-1}W\chi_1
\rangle_{\rm at},
$$
where $F = \overline{\boldsymbol{\chi}} (H_0 -z)^{-1}
\overline{\boldsymbol{\chi}}$ is a function of $H_f$, that is, $F=F(H_f)$ with
\begin{eqnarray} \label{eq:fterms}
F(r) := \frac{\overline{\boldsymbol{\chi}}^2(s,r)}{ H_{\rm at}(s)-z + r},
\end{eqnarray}
and $\overline{\boldsymbol{\chi}}(s,r) = \overline{P}_{\rm at}(s)
\otimes 1 + P_{\rm at}(s) \otimes \overline{\chi}_1(r)$.
Since $W = a(G) + a^*(G)$, the $L$th term in
this series is a sum of $2^L$ terms. We label them by  $L$-tuples
$\underline{\sigma} = (\sigma_1, \sigma_2,  ... ,\sigma_L)$, with
 $\sigma_i \in \{ - , + \}$, and we set $a^+(G) :=a^*(G)$, $a^-(G) :=
 a(G)$. With these notations
\begin{equation}  \label{eq:explicit}
\langle \chi_1 (W F)^{L-1} W  \chi_1 \rangle_{\rm at} =
\sum_{\underline{\sigma} \in \{-,+\}^{L}}
 \langle \chi_1
 \prod_{j=1}^{L-1} \left\{ a^{\sigma_j}(G) F(H_f) \right\}
a^{\sigma_L}(G) \chi_1 \rangle_{\rm at}.
\end{equation}
Next we use a variant of Wick's theorem (see \cite{bacfrosig:ren}) to expand each term of the sum
\eqref{eq:explicit} in a sum of normal ordered terms. Explicitly, this means
that in each term of \eqref{eq:explicit} the pullthrough formulas
$$
f(H_f) a^*(k) = a^*(k) f( H_f + |k| ),
 \qquad a(k) f(H_f)  = f(H_f + |k| ) a(k),
$$
and the canonical commutation relations are used to move all creation operators
to the very left, and all annihilation operators to the right of all
other operators. To write down the result we introduce the multi-indices
\begin{eqnarray*}
\underline{m,p,n,q} := (m_1,p_1,n_1,q_1,\ldots, m_L,p_L,n_L,q_L) \in
\{0,1\}^{4L} \; ,
\end{eqnarray*}
which run over the sets
$I_L := \{ \underline{m,p,n,q} \in  \{0,1\}^{4L}|m_l + p_l + n_l + q_l = 1\}$. The numbers $m_l,p_l,n_l,q_l$ may be thought
of as flags that indicate the position of the operator $a^{\sigma_l}(k)$ in a
given normal-ordered term: $m_l=1$ ($n_l=1$) if it is a \emph{non-contracted}
creation (annihilation) operator, $p_l=1$ ($q_l=1$) if it is a \emph{contracted}
creation (annihilation) operator. We obtain
\begin{eqnarray} \label{marcel:2}
\lefteqn{ \langle  \chi_1 (W F)^{L-1} W \chi_1 \rangle_{\rm at} }  \\
&& \nonumber = \sum_{\underline{m,p,n,q} \in I_L } \int d k_{\um}
d\tilde{k}_{\un} \left\{ \prod_{l=1}^L a^*(k_{m_l})^{m_l} \right\}
V_{\underline{m,p,n,q}}(H_f, k_{\um}, \tilde{k}_{\un} ) \left\{
\prod_{l=1}^L a(\tilde{k}_{n_l})^{n_l} \right\} \; ,
\end{eqnarray}
with
\begin{eqnarray}
\lefteqn{ V_{\underline{m,p,n,q}}(r, k_{\um},  \tilde{k}_{\un} )} \nonumber \\
&=&  \chi_1(r + r_0(\um,\un)) \Bigg\langle \left\{\prod_{l=1}^{L-1}
G(k_{m_l})^{m_l} G^*(\tilde{k}_{n_l})^{n_l} a^*(G)^{p_l}
 a(G)^{q_l}
F(H_f + r + r_l(\um, \un ) ) \right\}\nonumber  \\
&& \times G(k_{m_L})^{m_L} G^*(\tilde{k}_{n_L})^{n_L} a^*(G)^{p_L}
 a(G)^{q_l} \Bigg\rangle_{{\rm at}, \Omega} \chi_1(r+ r_L(\um,\un)),\label{marcel:3}
\end{eqnarray}
where $\langle A \rangle_{{\rm
at},\Omega} := ( \Omega , \langle A \rangle_{\rm at} \Omega)$, $\Omega\in\FF$
being the vacuum vector. Moreover
\begin{align*}
    k_\um &:= ( m_1 k_1,\ldots , m_L k_L),& d\tilde{k}_{\um} &:=
    \prod_{l=1,m_l=1}^{L}dk_{l}\\
    \tilde{k}_\un &:= ( n_1 \tilde{k}_1, \ldots , n_L \tilde{k}_L),& d\tilde{k}_{\un} &:=
    \prod_{l=1,n_l=1}^{L}d\tilde{k}_{l}
\end{align*}
and
$$
  r_l(\um, \un ) = \sum_{\stackrel{i \leq l}{m_l =1}} |k_i| +
  \sum_{\stackrel{i \geq  l + 1}{n_l = 1}} |\tilde{k}_i|.
$$
Upon summing \eqref{marcel:2} for $L=1$ through $\infty$ we
collect all terms with equal numbers $M =
|\underline{m}|:=\sum_{l=1}^L m_l$ and $N=|\underline{n} |:=\sum_{l=1}^L n_l$
of creation and annihilation operators, respectively.
To this end we need to relabel the integration variables. That is,
we distribute the $M+N$ variables $k_1,\ldots,k_M\in\R^3\times\{1,2\}$ and $\tilde{k}_1,\ldots,\tilde{k}_N\in\R^3\times\{1,2\}$
into the $M+N$ arguments of $V_{\underline{m,p,n,q}}(r,\cdot,\cdot)$ designated by $m_l=1$ and $n_l=1$.
Explicitly this is done by
$$
   \sigma_{\um}(k_1,\ldots,k_M)=( m_1 k_{\um(1)},\ldots , m_L
   k_{\um(L)}),\qquad \um(l)=\sum_{j=1}^l m_j.
$$
We obtain
\begin{eqnarray*}
   \lefteqn{\sum_{L\geq 1}(-1)^{L-1}g^L\langle  \chi_1 (W F)^{L-1} W \chi_1
    \rangle_{\rm at}}\\ &=& \sum_{M+N \geq 1}\int_{B^{M+N}} a^*(k^{(M)})
    \hat{w}_{M,N}(H_f,K)a(\tilde{k}^{(N)})\,dK
\end{eqnarray*}
where
\begin{equation}\label{eq:w-sum-V}
    \hat{w}_{M,N}(r,K) = \sum_{L \geq M+N} ( - 1)^{L-1} g^L
    \sum_{ \stackrel{\underline{m,p,n,q}\in I_L} {|\um|
    = M , |\underline{n}|= N}} V_{\underline{m,p,n,q}}
   (r,\sigma_{\um}(k^{(M)}),\sigma_{\un}(\tilde{k}^{(N)}))
\end{equation}
and $K=(k^{(M)},\tilde{k}^{(N)})$. Hence  $H^{(0)}_g=H(w)$ with
\begin{equation}\label{eq:w0-sum-V}
w_{0,0}(r) = E_{at} - z +r+ \sum_{L \geq 1} ( - 1)^{L-1} g^L \sum_{
\stackrel{\underline{p,q} \in \{0,1\}^{2L}}{p_l + q_l = 1 }}
V_{\underline{0,p,0,q} } (r ),
\end{equation}
and $w_{M,N}(r,K)$ given by the symmetrisation of $\hat{w}_{M,N}(r,K)$
with respect to $k_1,\ldots,k_M\in\R^3$ and $\tilde{k}_1,\ldots,\tilde{k}_N\in\R^3$, respectively.

It remains to show that $H(w)-(E_{\rm at}-z)$ belongs to the ball
$\BB(\alpha_0, \beta_0, \gamma_0)$ for $g$ sufficiently small. To
this end we need the following estimates on the operator-valued function \eqref{eq:fterms}
and on its derivative,
\begin{eqnarray} \label {eq:derivative} F'(r) = -
\frac{\overline{\bchi}^2(s,r)}{(H_{\rm at}(s) - z + r )^2} +
\frac{P_{\rm at}(s) \otimes 2  \overline{\chi}_1(s,r)\partial_r
\overline{\chi}_1(s,r)}{H_{\rm at}(s) -z + r} \; .
\end{eqnarray}

\begin{lemma} \label{lem:technicalities}
Let Hypothesis \ref{hyp:G} and \ref{hyp:R} hold for some $\mu
> 0$ and $\UU$. Then
\begin{eqnarray*}
C_0 &: =& \sup_{(s,z) \in \UU} \sup_{r \geq 0} \| (H_f + 1) F(H_f +
r ) \|  < \infty
\\
C_1 &: =& \sup_{(s,z) \in \UU} \sup_{r \geq 0} \|
 (H_f + 1) F'(H_f + r ) \|  < \infty \; ,
\end{eqnarray*}
for $F $  given by  \eqref{eq:fterms}.
\end{lemma}

\begin{proof}
To show that $C_0$ is finite we estimate
\begin{eqnarray*}
\lefteqn{ \sup_{r \geq 0} \left\| ( H_f + 1 )
\frac{\overline{\bchi}^2(s, H_f + r)}{H_{\rm at}(s) - z + H_f + r}
\right\|} \\
&& = \sup_{r,q \geq 0} \left\| ( q + 1 ) \frac{\overline{P}_{\rm
at}(s) \otimes 1 + P_{\rm at}(s)\otimes
\overline{\chi}_1^2(r + q)}{H_{\rm at}(s) - z + q + r} \right\| \\
&& \leq \sup_{r,q \geq 0} \left\| ( q + 1 ) \frac{\overline{P}_{\rm at}(s)}{H_{\rm at}(s) - z + q + r} \right\|  +
\sup_{r,q \geq 0} \left\| ( q + 1 ) \frac{\overline{\chi}_1^2 (r +
q)}{E_{\rm at}(s) - z + q + r}
 \right\| \| P_{\rm at}(s) \|.
\end{eqnarray*}
By Hypothesis \ref{hyp:R}, both terms are bounded on $\UU$.
Similarly $C_1$ is estimated using \eqref{eq:derivative}.
\end{proof}


\begin{proof}[\it Proof of Theorem~\ref{thm:sigal}]
Let Hypothesis \ref{hyp:G} and \ref{hyp:R} hold for some $\mu
> 0$ and $\UU$. Let $0< \xi < 1$. By  Theorem \ref{thm:fesh} we know that
there exists a $g_0 > 0$ such that for all $|g| < g_0$, $(H_g -z,H_0
- z)$ on $\UU$ is a Feshbach pair for $\boldsymbol{\chi}$. Let
$(s,z) \in \UU$. First we derive upper bounds for
$V_{\underline{m,p,n,q}}$ and $\partial_r V_{\underline{m,p,n,q}}$.
Inserting $(H_f+1)^{-1}(H_f+1)$ in front of $F(H_f + r +  r_l(\um, \un ) )$ we
obtain, from Lemma~\ref{lem:technicalities}, that
\begin{eqnarray}
\lefteqn{| V_{\underline{m,p,n,q}}(r,k_{\um},\tilde{k}_{\un} )|}
\nonumber \\
&\leq &\left\{\prod_{l=1}^{L}
   \|G( k_{m_l})\|^{m_l}\|G( k_{n_l})\|^{n_l}
   \|G\|_{\omega}^{p_l+q_l} \right\}
   C_0^{L-1}\sup_{s:(s,z)\in \UU}\|P_{\rm at}(s)\| \label{V1}.
\end{eqnarray}
Let $C_{\rm at} := \sup_{s:(s,z)\in \UU} \|
P_{\rm at}(s) \|$. Similarly, using \eqref{marcel:3}, \eqref{eq:derivative}
and \eqref{V1} we estimate
\begin{eqnarray}
 \lefteqn{|\partial_r V_{\underline{m,p,n,q}}(r,k_{\underline{m}},\tilde{k}_{\underline{n}})|}\\
 &\leq & 2\|\chi_1'\|_\infty  \cdot \left\{\prod_{l=1}^{L}
   \|G( k_{m_l})\|^{m_l}\|G( k_{n_l})\|^{n_l}
   \|G\|_{\omega}^{p_l+q_l} \right\}
   C_0^{L-1}C_{\rm at} \nonumber \\
&& +\sum_{j=1}^{L-1}  \Bigg\vert\Bigg\langle\left\{\prod_{l=1}^{j-1}
G( k_{m_l})^{m_l} G^*( \tilde{k}_{n_l})^{n_l} a^*(G_g)^{p_l}
 a(G)^{q_l}
F(H_f +  r + r_l(\um, \un )) \right\} \nonumber \\
&& \qquad\times G( k_{m_j})^{m_j} G^*( \tilde{k}_{n_j})^{n_j} a^*(G)^{p_j}
 a(G)^{q_j}
  {F}'(H_f + r + r_j(\um, \un ) )  \nonumber \\
&& \qquad\times \left\{\prod_{l=j+1}^{L-1} G (k_{m_l})^{m_l} G^*(\tilde{k}_{n_l})^{n_l} a^*(G)^{p_l}
 a(G)^{q_l}F(H_f + r +  r_l(\um,\un) ) \right\}  \nonumber \\
&& \qquad\times G(k_{m_L})^{m_L} G^*(\tilde{k}_{n_L})^{n_L} a^*(G)^{p_L}
 a(G_g)^{q_l}   \Bigg\rangle_{{\rm at}, \Omega} \Bigg\vert
 \nonumber \\
&\leq & \left\{\prod_{l=1}^{L}
   \|G( k_{m_l})\|^{m_l}\|G( k_{n_l})\|^{n_l}
   \|G\|_{\omega}^{p_l+q_l} \right\} C_{\rm at}  C_0^{L-2} ( 2 \| \chi_1' \|_\infty    C_0 + (L-1)C_1 )  .
\label{V2}
\end{eqnarray}
With the help of \eqref{V1} and \eqref{V2} we can now prove the theorem.
From \eqref{eq:w0-sum-V} and \eqref{V1} it follows that
\begin{eqnarray*}
\left| w_{0,0}(0) -  (E_{\rm at} - z ) \right| &\leq &\sum_{L=2}^\infty g^L
\sum_{\stackrel{\underline{p,q} \in \{0,1\}^{2L} }{p_l+ q_l = 1}}
\left| V_{\underline{0,p,0,q}}(0) \right|\\
&\leq & \sum_{L=2}^\infty g^L \sum_{\stackrel{\underline{p,q} \in
\{0,1\}^{2L} }{p_l+ q_l = 1}} \|G\|_\omega^L C_0^{L-1} C_{\rm at}\\
&\leq & C_{\rm at} \sum_{L=2}^\infty 2^L g^L \| G \|_\omega^L C_0^{L-1},
\end{eqnarray*}
which can be made smaller than any positive $\alpha_0$ for small $g$. Estimate
\eqref{V2} implies that
\begin{eqnarray*}
\| {w}'_{0,0} - 1 \|_\infty
&=& \sup_r | {{w}'}_{0,0}[r] - 1 | \\
&\leq & \sup_r \sum_{L=2}^\infty g^L
\sum_{\stackrel{\underline{p,q} \in \{ 0, 1 \}^{2L}}{p_l+q_l=1}}
| \partial_r V_{0,p,0,q}(r)|\\
&\leq &  \sum_{L=2}^\infty g^L \sum_{ \stackrel{\underline{p,q} \in
\{0,1\}^{2L}}{p_l + q_l = 1}} \| G\|_\omega^L C_{\rm at}
C_0^{L-2} \left( 2 C_0 \|\chi_1' \|_\infty  + (L-1)
{C_1} \right) \\
&\leq &  \sum_{L=2}^\infty g^L \| G \|_\omega^L 2^L C_{\rm at}
C_0^{L-2} \left( 2 C_0 \| \chi_1' \|_\infty + (L-1) {C_1} \right),
\end{eqnarray*}
which can be made  smaller than any positive $\beta_0$ for small $g$. It remains to
show that $\|\left( {w}_{M,N} \right)_{M+N \geq 1}\|_{\mu,\xi}\leq\gamma_0$ for
$g$ sufficiently small. By \eqref{eq:w-sum-V}
\begin{equation}\label{V3}
    \|w_{M,N}\|_{\mu} \leq \sum_{L\geq M+N} g^L \sum_{ \stackrel{\underline{m,p,n,q}  \in I_L } { |
\underline{m} | = M, | \underline{n} | = N } } \|V_{\underline{m,p,n,q}}\|_{\mu}
\end{equation}
where, by a triangle inequality and by \eqref{V1} and \eqref{V2}
\begin{eqnarray}\nonumber
   \|V_{\underline{m,p,n,q}}\|_{\mu} &\leq
   &\left(\int_{B^{M+N}}\|V_{\underline{m,p,n,q}}(K)\|_{\infty}^2\frac{dK}{|K|^{2+2\mu}}\right)^{1/2}\\
   && + \left(\int_{B^{M+N}}\|\partial_r
   V_{\underline{m,p,n,q}}(K)\|_{\infty}^2\frac{dK}{|K|^{2+2\mu}}\right)^{1/2}\nonumber\\
   &\leq & \|G\|_\mu^{M+N} \|G\|_{\omega}^{L-(M+N) } S_L, \label{V4}
\end{eqnarray}
with $ S_L := C_{\rm at}  C_0^{L-2} \left( C_0 + 2 \| \chi_1'
\|_\infty C_0 + (L-1)C_1 \right)$, and
$$
    \|G\|_\mu := \left(\int_{\R^3}\|G(k)\|^2\frac{dk}{|k|^{2+2\mu}}\right)^{1/2}.
$$
Combining \eqref{V3} and \eqref{V4} and find
\begin{eqnarray*}
\| {w}_{M,N} \|_\mu \leq \sum_{L=1}^\infty g^L S_L \sum_{
\stackrel{\underline{m,p,n,q}  \in I_L } { | \underline{m} | = M, |
\underline{n} | = N } } \| G \|_\mu^{M+N} \| G \|_\omega^{L-(M+N)}
\; ,
\end{eqnarray*}
where the condition $L\geq M+N$ has been relaxed to $L\geq 1$. Therefore
\begin{eqnarray*}
\| \left( {w}_{M,N} \right)_{M+N \geq 1} \|_{\mu,\xi}
&=& \sum_{M+N \geq 1} \xi^{-(M+N)} \| {w}_{N,M} \|_\mu \\
&\leq& \sum_{L=1}^\infty g^L S_L \| G \|_\omega^L \sum_{M+N \geq 1}
\xi^{-(M+N)} \sum_{ \stackrel{\underline{m,p,n,q}  \in I_L } { |
\underline{m} | = M, | \underline{n} | = N } } \left( \| G
\|_\omega^{-1} \| G \|_\mu  \right)^{M+N}
 \\
&\leq& \sum_{L=1}^\infty g^L S_L \| G \|_\omega^L \left(
 \sum_{ \underline{m,p,n,q}  \in I_1 }
\left(
\xi^{-1} \| G \|_\omega^{-1} \| G\|_\mu \right) ^{m_1+n_1} \right)^L \\
&\leq&
 \sum_{L=1}^\infty g^L S_L \| G \|_\omega^L \left(2  + 2 \xi^{-1} \| G \|_\omega^{-1} \| G \|_\mu \right)^L \; .
\end{eqnarray*}
This can be made  smaller than any positive $\gamma_0$ for small
coupling $g$. It follows that we can find a $g_1>0$ such that on
$\UU$, \eqref{eq:nb1} holds for all $g\in[0,g_1)$. This concludes the
proof.
\end{proof}


\section{Technical Auxiliaries}
\label{sec:aux}

Let $L^2(\R^3\times\{1,2\},\LL(\HH_{\rm at}))$ be the Banach space of
(weakly) measurable functions $T:\R^3\times\{1,2\}\to\LL(\HH_{\rm at})$ with
$\int\|T(k)\|^2 dk<\infty$, and let
$$
   \|T\|_\omega := \left(\int\|T(k)\|^2(|k|^{-1}+1)dk \right)^{1/2}.
$$

\begin{lemma} \label{lem:elemesitmates}\label{lm:A1} If
$T\in L^2(\R^3\times\{1,2\},\LL(\HH_{\rm at}))$, then
\begin{eqnarray*}
\| a(T)(H_f + 1 )^{-1/2} \| &\leq &\left( \int\|T(k)\|^2 |k|^{-1} dk \right)^{1/2},  \\
\| a^*(T)(H_f + 1 )^{-1/2} \| &\leq &\|T\|_\omega. \\
\end{eqnarray*}
\end{lemma}
For a proof of this lemma see, e.g., \cite{BFS1}.

\begin{lemma} Suppose the function
$F : U  \to  \mathcal{L}(\HH_{\rm at}, L^2(\R^3 ; \HH_{\rm at}))$,
$s \mapsto F_s$ is uniformly bounded and suppose for a.e. $k \in
\R^3$ and all $s \in U$, there exists an operator
$F_s(k) \in \mathcal{L}(\HH_{\rm at})$ such that $F_s(k) \varphi =
(F_s \varphi)(k)$ for all $\varphi \in \HH_{\rm at}$. If for a.e. $k
\in \R^3$, the function $s \mapsto F_s(k) \in \mathcal{L}(H_{\rm
at})$ is analytic, then $F$ is analytic.
\end{lemma}
\begin{proof}
Let $h \in L^2(\R^3)$ and  $\varphi_1, \varphi_2 \in \HH_{\rm at}$,
and suppose $\gamma$ is a nullhomotopic closed curve in $U$. Then
$$
\int_\gamma ( h \otimes \varphi_1 , F_s \varphi_2 ) ds = \int_\gamma
 \int \overline{h}(k) ( \varphi_1, F_s(k) \varphi_2) dk  ds
 = \int \overline{h}(k) \int_\gamma ( \varphi_1,  F_s(k) \varphi_2)
 d s dk = 0 \; ,
$$
where we interchanged the order of integration, which is  justified
since $F$ is uniformly bounded. It follows that $ s \mapsto ( h
\otimes \varphi_1, F_s \varphi_2 )$ is analytic. By linearity we
conclude that $s \mapsto ( \psi, F_s \varphi_2 )$ is analytic for
all $\psi$ in a dense linear subset of $\HH_{\rm at} \otimes
\mathfrak{h}$. This and the uniform boundedness imply strong
analyticity, see for example the remark following Theorem
3.12 of Chapter III in \cite{kat:per}.
\end{proof}

\begin{prop} \label{pro:analytredres} Let $R \ni s \mapsto
T(s)$ be an analytic family. Suppose there exists an isolated
non-degenerate eigenvector $E(s)$ with analytic projection operator
$P(s)$. Let $\overline{P}(s) := 1 - P(s)$ and let
\begin{align*}
  \Gamma := \{ (s,z) \in R \times \mathbb{C} \ | &\ (T(s) - z)\
  \text{is a bijection from}\ D(T(s))\cap\ran\overline{P}(s)\ \text{to}\
  \ran\overline{P}(s)\\
  &\text{with bounded inverse}\}.
\end{align*}
Then $\Gamma$ is open and $(s,z)\mapsto(T(s) - z)^{-1}\overline{P}(s)$ is
analytic on $\Gamma$.
\end{prop}

\begin{proof} Let $(s_0,z_0) \in \Gamma$. There exists in a
neighborhood of $s_0$ a bijective operator $U(s): \HH \to \HH$,
analytic in $s$, such that $U(s) P(s) U(s)^{-1} = P(s_0)$ and hence
$U(s) \overline{P}(s) U(s)^{-1} = \overline{P}(s_0)$
(\cite{reesim:ana} Thm. XII.12). The operator $\widetilde{T}(s) =
U(s) T(s) U(s)^{-1}$ is an analytic family. It leaves $\ran
\overline{P}(s_0)$ invariant and thus $\widetilde{T}(s) \restricted
\ran \overline{P}(s_0) : \ran \overline{P}(s_0) \cap
\mathcal{D}(\widetilde{T}(s)) \to \ran \overline{P}(s_0)$ is an
analytic family as well. By this and the fact that
$(\widetilde{T}(s_0) - z_0 )\restricted \ran \overline{P}(s_0)$  is
bijective with bounded inverse since $(s_0,z_0) \in \Gamma$, it
follows by \cite{reesim:ana} Thm. XII.7 that in a neighborhood of
 $(s_0,z_0)$, $(\widetilde{T}(s) - z ) \restricted
\ran \overline{P}(s_0)$ is bijective with bounded inverse and
$(\widetilde{T}(s) - z )^{-1} \overline{P}(s_0)$ is analytic in both
variables. Thus in this neighborhood also the function $(T(s)-z)
\restricted \overline{P}(s)$ is bijective with bounded inverse and $(T(s) -
z )^{-1}\overline{P}(s)$ is an analytic function of two variables.
\end{proof}

\begin{theorem} Suppose the assumptions of Lemma \ref{prop:balls} hold.
Then in the norm  of $\mathcal{L}(\HH_{\rm red})$,
$$
\lim_{n \to \infty} H^{(n)}(z_\infty) = \lambda H_f  \; .
$$
for some $\lambda\in\C$.
\end{theorem}

\begin{proof}
We recall the notations $H^{(n)}(z_\infty) = H(w^{(n)}(z_\infty))$
and  $E^{(n)}(z_\infty) = w^{(n)}_{0,0}(z_\infty,0)$. Using the
decomposition
$$
H^{(n)}(z_\infty) =  \left( H^{(n)}(z_\infty) -
w^{(n)}_{0,0}(z_\infty,H_f) \right) + \left(
w^{(n)}_{0,0}(z_\infty,H_f) - E^{(n)}(z_\infty) \right) +
E^{(n)}(z_\infty)\; ,
$$
the theorem will follow from Steps~1 and 2 below.

\bigskip
\noindent\underline{Step 1}: $\lim_{n \to \infty}
\| H^{(n)}(z_\infty) - w^{(n)}_{0,0}(z_\infty,H_f)  \| = 0$ and
$\lim_{n\to \infty} E^{(n)}(z_\infty) = 0$.

From Lemma~\ref{cor:bcfs} we know that 
\begin{eqnarray} \label{eq:convv5}
H^{(n)}(z) - \rho^{-1} E^{(n-1)}(z) \in
\BB(\alpha_n,\beta_n,\gamma_n) \; ,
\end{eqnarray}
for  $z \in U_n$. By \eqref{eq:westimate} this implies that
$$
\| H^{(n)}(z_\infty) - w^{(n)}_{0,0}(z_\infty,H_f)  \|  \leq \|
w^{(n)}(z_\infty) - w^{(n)}_{0,0}(z_\infty)  \|_{\mu,\xi} \leq
\gamma_n  \to 0 \quad (n \to \infty).
$$

By \eqref{eq:convv5},
\begin{eqnarray} \label{eq:convvv1}
| E^{(n)}(z) | \leq \rho \alpha_{n+1} + \rho | E^{(n+1)}(z) | \quad
, \quad z \in U_n  \; .
\end{eqnarray}
Iterating \eqref{eq:convvv1}, we find
$$
|E^{(n)}(z) | \leq \sum_{k=1}^m \rho^k \alpha_{n +k} + \rho^m |
E^{(n+m)}(z)|,
$$
which yields,
$$
|E^{(n)}(z_{n+m}) | \leq \sum_{k=1}^\infty \rho^k \alpha_{n+k}.
$$
Since $E^{(n)}$ is continuous and $\lim_{n \to \infty} z_n =
z_\infty$, we arrive at
$$
|E^{(n)}(z_\infty)| \leq \sum_{k=1}^{\infty} \rho^k \alpha_{n+k} \to
0 \quad , \quad  (n \to \infty) \; .
$$

\bigskip
\noindent \underline{Step 2}:
There exists a $\lambda\in\C$ such that
$$
\lim_{n \to \infty} ( w^{(n)}_{0,0}(z_\infty , r) -
w^{(n)}(z_\infty,0)) = \lambda r \; ,
$$
 uniformly in $0 \leq r \leq 1$.
\vspace{0.5cm}

To abbreviate the notation, we set $T^{(n)}(z_\infty,r) :=
w^{(n)}_{0,0}(z_\infty , r) - w^{(n)}(z_\infty,0)$.
 From \cite{bacchefrosig:smo} (3.105-3.107), we have
\begin{eqnarray} \label{eq:convv1}
T^{(n)}(z_\infty, r ) = \rho^{-1} T^{(n-1)}(z_\infty,\rho r) +
e^{(n-1)}(z_\infty, r) \; ,
\end{eqnarray}
with $ e^{(n-1)}(z_\infty,0) =0$  and
\begin{eqnarray} \label{eq:convv2}
 \sup_{r \in
[0,1]} \Big(|\partial_r e^{(n)}(z_\infty,r)|  +
 | e^{(n)}(z_\infty,r)|\Big) \leq C \gamma_n^2.
 \end{eqnarray}
Iterating \eqref{eq:convv1}, we arrive at
\begin{eqnarray} \label{eq:convv3}
T^{(n)}(z_\infty,r) = \rho^{-n} T^{(0)}(z_\infty,\rho^n r) + \sum_{k=0}^{n-1}
 \rho^{-(n-1-k)} e^{(k)}(z_\infty,\rho^{n-1-k} r) \; .
\end{eqnarray} 
To prove Step 2 we now show that, uniformly in $r \in [0,1]$,
$$
\lim_{n \to \infty} T^{(n)}(z_\infty, r) = r  \left( \partial_r
T^{(0)}(z_\infty,0) + \sum_{k=0}^\infty \partial_r e^{(k)}(z_\infty,
0) \right).
$$
Note that the series on the right hand side
converges by \eqref{eq:convv2}. Given $\epsilon > 0$ we choose $K$ so large that
\begin{equation} \label{eq:convv4}
 \sum_{k=K}^\infty C \gamma_k^2 \leq
\epsilon.
\end{equation}
By \eqref{eq:convv3} and the triangle inequality, we find for $n
\geq K$, (suppressing $z_\infty$)
\begin{eqnarray*}
\lefteqn{ \left| T^{(n)}( r) - r  \left( \partial_r T^{(0)}(0) -
\sum_{k=0}^\infty \partial_r e^{(k)}(
0) \right) \right| } \\
&\leq& \left| \rho^{-n} T^{(0)}(\rho^n r ) - r \partial_r T^{(0)}(
0) \right| + \sum_{k=0}^K \left| \rho^{-(n-1-k)}
e^{(k)}(\rho^{n-1-k} r) - r
\partial_r e^{(k)}(0) \right| \\
&& + \sum_{k=K+1}^\infty | \rho^{-(n-1-k)} e^{(k)}(\rho^{n-1-k} r )
| + \sum_{k=K+1}^\infty | r \partial_r e^{(k)}(0)| \; .
\end{eqnarray*}
The first two terms on the right hand side converge to zero as $n$
tends to infinity because $T^{(n)}(0)=0$  and $e^{(k)}(0)=0$. The
last term on the right hand side is bounded by $\epsilon$, which
follows from Eqns. \eqref{eq:convv2} and \eqref{eq:convv4}. Using
again \eqref{eq:convv2} and \eqref{eq:convv4} we see that the first
term on the last line is bounded by $\epsilon$ as well, since, by the mean
value theorem, $ \alpha^{-1} | e^{(n)}(\alpha r) | \leq \sup_{\xi \in
[0,1]} | {e^{(n)}}'(\alpha \xi )|  r$ for $\alpha,r \in [0,1]$.
\end{proof}

\end{document}